\documentclass[11pt]{article}

\usepackage[margin=1.1in]{geometry}
\usepackage{amsmath,amssymb,amsfonts,amsthm,mathtools}
\usepackage{mathrsfs}
\usepackage{graphicx}
\usepackage{flafter}
\usepackage{bm}
\usepackage{booktabs,tabularx,array,enumitem}
\usepackage{algorithm}
\usepackage{algpseudocode}
\algrenewcommand{\algorithmicrequire}{\textbf{Input:}}
\algrenewcommand{\algorithmicensure}{\textbf{Output:}}
\usepackage{multirow}
\usepackage[authoryear,round]{natbib}
\usepackage{xcolor}
\usepackage{microtype}
\usepackage{hyperref}
\usepackage{cleveref}
\crefname{figure}{Fig.}{Figs.}
\Crefname{figure}{Fig.}{Figs.}
\usepackage{titlesec}
\usepackage{fancyhdr}
\usepackage{parskip,placeins}
\usepackage{subcaption}
\crefname{subfigure}{Fig.}{Figs.}
\Crefname{subfigure}{Fig.}{Figs.}
\usepackage{tikz,pgfplots}
\usetikzlibrary{arrows.meta,positioning,bayesnet}
\newcommand{\cmsedge}[2]{%
  \foreach \cmsparent in {#1}{%
    \foreach \cmschild in {#2}{%
      \draw[->] (\cmsparent) -- (\cmschild);
    }%
  }%
}
\pgfplotsset{compat=1.18}

\tikzset{
  cmsnode/.style={minimum size=10mm},
  cmschange/.style={draw=orange!80!black,line width=1.2pt},
  cmsnote/.style={font=\small,align=center},
  cmsenv/.style={draw=orange!80!black,rectangle,fill=orange!8,
                minimum size=7mm,inner sep=2pt,line width=.8pt},
  cmsshift/.style={->,dashed,draw=orange!80!black,line width=1pt},
  cmsinput/.style={draw=blue!55!black,dashed,rounded corners}
}

\newtheorem{proposition}{Proposition}
\newtheorem{corollary}{Corollary}
\newtheorem{lemma}{Lemma}
\newtheorem{definition}{Definition}

\hypersetup{
  colorlinks=true,
  linkcolor=blue!70!black,
  citecolor=blue!50!black,
  urlcolor=blue!60!black,
  pdftitle={Conformal Prediction with Stable Measurement Mechanisms},
  pdfauthor={Seungjin Choi}
}

\providecommand{\E}{\mathbb{E}}
\providecommand{\Prob}{\mathbb{P}}

\newcommand{\N}{\mathcal N}
\newcommand{\ind}{\mathbf 1}
\newcommand{\pa}{\operatorname{pa}}
\newcommand{\KL}{\operatorname{KL}}
\newcommand{\TV}{d_{\mathrm{TV}}}
\newcommand{\Var}{\operatorname{Var}}

\newcommand{\calD}{\mathcal D}
\newcommand{\calJ}{\mathcal J}
\newcommand{\calC}{\mathcal C}
\newcommand{\calA}{\mathcal A}

\newcommand{\calY}{\mathcal Y}
\newcommand{\dd}{\,\mathrm d}

\DeclareMathOperator*{\argmax}{arg\,max}

\titleformat{\section}{\large\bfseries\color{blue!60!black}}{\thesection}{1em}{}[\titlerule]
\titleformat{\subsection}{\normalsize\bfseries\color{blue!40!black}}{\thesubsection}{1em}{}
\titleformat{\subsubsection}{\normalsize\itshape\bfseries}{\thesubsubsection}{1em}{}

\begin{document}

\begin{center}
  {\LARGE\bfseries \textsf{Conformal Prediction\\[0.5em]
    with Stable Measurement Mechanisms}} \\[2em]
  {\large Seungjin Choi}\\[1em]
  {\normalsize CROID Research and aSSIST University, Seoul, Korea}\\[1em]
\end{center}

\vspace{0.5em}
\noindent\rule{\linewidth}{1.5pt}
\vspace{0.5em}

\begin{abstract}
In applications such as air-quality monitoring, inexpensive
sensor readings can help predict pollutant concentrations
measured by a separate reference instrument.
We study how to adapt conformal prediction intervals to a
new population when reference outcomes are available in the
source data, but only sensor readings and context, such as
weather conditions, are available in the target.
We assume that the measurement distribution at a given
outcome and context remains stable, while context frequencies
and outcome distributions may change.
This model is a particular form of causal mechanism shift.
Its statistical restriction has been studied as sparse joint shift.
Using established marginal matching, predictive adjustment,
and weighted conformal calibration, we examine what this
restriction contributes to coverage and interval length
for continuous outcomes.
Gaussian analyses show how measurement stability can identify
outcome changes and how retaining context can shorten prediction intervals.
For Gaussian models with known source mechanisms, we derive
exact coverage formulas in the large-calibration limit while
keeping the unlabeled target sample finite.
These formulas explain when sharing a correction across
contexts is preferable to fitting separate corrections,
and why estimating outcome variability is particularly
sensitive to weakly informative measurements.
Experiments on synthetic Gaussian data show a coverage gain over
covariate weighting at the same mean length and a length reduction
of about $8\%$ relative to pooled label adjustment at comparable
coverage. On the Air
Quality dataset, controlled shifts yield coverage gains with
longer intervals, while spring-to-autumn and year-over-year transfers
show substantial undercoverage.
\end{abstract}

\newpage
\section{Introduction}
\label{sec:intro}

Distribution shift often arises because some parts of a
data-generating process change while others remain stable.
Causal mechanism-shift models describe this structure by
specifying which conditional distributions, given their causal
parents, may change across populations
\citep{PetersJ2017book,PearlJ2014ss,SubbaswamyA2022jci}.
For prediction, this structure matters because intervals calibrated
in one population may lose coverage in another. When outcomes are
unavailable in the target population, changes in the observed
inputs alone generally do not determine how predictions should
be corrected. Stable parts of the process can provide the
additional information needed for adaptation.

We focus on a particular mechanism-shift model with a stable
measurement process. A measurement mechanism specifies the
distribution of readings given the quantity being measured and the
measurement conditions. This distribution includes measurement noise.
Stability means that this full conditional distribution is the same in the
source and target populations. The marginal distribution of readings
may still change as the population changes.

Air-quality monitoring provides a concrete example. The outcome
$Y$ is the pollutant concentration reported by a reference analyzer.
The measurements $Z$ come from separate, inexpensive sensors and
may record electrical responses rather than concentration itself.
Weather conditions, such as temperature and humidity, form the
context $C$. Source data contain paired observations of $(C,Z,Y)$,
whereas adaptation uses target observations of $(C,Z)$ without
reference outcomes. The goal is to predict $Y$ from $X=(C,Z)$ in
the target population. Pollution levels and weather patterns can
change while the sensors retain their response distribution at a
fixed concentration and weather condition. Similarly, a new
industrial operating regime can change the pressures encountered
without changing a sensor's response at a given pressure and
ambient temperature.

We therefore allow context frequencies and outcome distributions
within each context to change, while assuming
\[
 p(z\mid y,c)=q(z\mid y,c),
\]
where $p$ and $q$ denote the source and target densities.
An informative measurement can then help infer outcome changes
from the unlabeled target inputs. Crucially, a stable sensor
response does not imply a stable prediction of the outcome from
the sensor. When the outcome distribution changes, the same
reading can support a different predictive distribution.

Our question is whether using this stability assumption improves
conformal prediction intervals relative to covariate-shift and
label-shift adjustment, and how reliably the required correction
can be estimated. Covariate shift keeps the outcome distribution
given the inputs fixed. Label shift keeps the distribution of all
inputs given the outcome fixed. Our model keeps only the
measurement distribution given outcome and context fixed, allowing
outcome changes to differ across contexts. This flexibility can
help adaptation, but its value depends on the information carried
by the measurements and the accuracy of the estimated correction.

The formulation also applies when $Z$ is a more general observed
proxy that provides information about $Y$ given $C$.
It need not come from a physical instrument or share the outcome's
units. In every application, conditional stability requires
justification. Sensor aging, recalibration, or unmeasured influences
can change the measurement process.

Statistically, our restriction is sparse joint shift (SJS)
\citep{ChenL2022neurips,TascheD2024fods}.
SJS allows the outcome and selected features to change jointly,
while preserving the distribution of the remaining features given
them. Here, context forms the selected feature block and the proxy
forms the remaining block. These papers establish identification,
ratio estimation, and predictive correction for categorical
outcomes. We study prediction intervals for continuous outcomes,
with the context and proxy roles specified before fitting.
The causal interpretation motivates these roles. The analysis
depends on the resulting statistical restriction.
Related work uses sparse mechanism changes to learn causal
structure \citep{PerryR2022neurips}. Here, the graph and the
mechanisms allowed to change are supplied as modeling assumptions.

The same density-ratio and marginal-matching principles apply to both classification and regression. 
For continuous outcomes, however, recovering the outcome distribution
from measurements can be unstable, even at a fixed context.
Distinct outcome distributions
may produce very similar measurement distributions, making them
difficult to distinguish from a finite target sample. We consider
specified families of outcome changes and examine the consequences
of this estimation problem for coverage and interval length.
The underlying identification requirement is established in
measurement-error and proxy methods
\citep{HuY2008econ,MiaoW2018biometrika}.

Estimation and calibration also build on existing methods.
Marginal matching fits a candidate outcome correction by comparing
its implied input distribution with the observed target inputs.
It appears in label-shift estimation
\citep{SaerensM2002neco,duPlessisM2012icml,LiptonZ2018icml,GargS2020neurips},
target-shift adaptation for both classification and regression
\citep{ZhangK2013icml}, SJS, and structured exponential tilting
\citep{MaityS2023iclr}.
Concepts and proxies also support adaptation under latent subgroup
shift \citep{AlabdulmohsinI2023aistats,TsaiK2024aistats}.
Weighted conformal calibration provides a coverage guarantee when
the density ratio is known
\citep{TibshiraniR2019neurips,PodkopaevA2021uai}. Estimating the
ratio introduces an additional approximation error
\citep{LeiL2021jrsssb}.
Our earlier work studies conformal adaptation under label shift
\citep{Choi2026eiml}, predictive tilting in
full conformal Bayes \citep{KimJY2026copa}, censored Gaussian
regression \citep{Choi2026arxiv_scbc}, sensitivity analysis \citep{Choi2026arxiv_jtscb}, exponential-tilt joint
shift \citep{Choi2026arxiv_extracp}, and conformal prediction for
molecular properties \citep{LeeHS2025ReliableML,LeeHS2026eiml}.

Building on this framework, we analyze what the stable measurement
restriction contributes to prediction intervals and when its
benefits survive estimation. The supporting identification and
calibration results and the oracle interval-length benchmark
are applications of established principles, credited alongside
their statements. Our contribution
is the explicit analysis of coverage after fitting Gaussian
corrections, its implications for using context, and the empirical
comparisons. The study has three parts.

\begin{enumerate}[leftmargin=*,itemsep=0.4em,topsep=0.4em]
\item \textbf{Prediction benefit of context and cost of separate corrections.}
A Gaussian example illustrates how the stability restriction
identifies an outcome correction that target inputs alone leave
ambiguous. An interval-length comparison with the target
distribution known shows when retaining context gives shorter
intervals, using classical minimum-volume reasoning
\citep{ScottC2006jmlr}.
We then distinguish using context for prediction from estimating
a separate shift in each context. Our coverage calculation shows
when the benefit of separate corrections outweighs their extra
estimation error
(\Cref{prop:contextvalue,prop:sharedcorrection}).

\item \textbf{Coverage loss from estimating outcome changes.}
For Gaussian models with known source mechanisms, we derive exact
coverage formulas in the large-calibration limit while keeping the
unlabeled target sample finite
(\Cref{prop:scaleexact,prop:scalecost,cor:contextcost}).
They quantify why estimating outcome variability is especially
sensitive to weak proxy information and how the fitting cost
changes across contexts. The scalar calculation also applies to
label shift, the no-context special case.

\item \textbf{Prediction comparisons under controlled and observed shifts.}
Experiments with estimated source mechanisms compare coverage and
length with covariate weighting and pooled label-shift adjustment.
The pooled Gaussian comparator uses $(Z,Y)$ and is correctly
specified, so its comparison isolates the value of retaining
context. One five-sensor Air Quality pipeline evaluates both
controlled shifts, constructed to preserve the measurement mechanism,
and observed calendar-based shifts. The controlled selection statistic
differs from the fitted tilt statistic.
\end{enumerate}

\section{Mechanism shift models and identification}
\label{sec:setup}
\label{sec:id}

This section describes the prediction data and the stable
measurement assumption. We first explain how marginal matching
fits an outcome correction through the sensor distribution
it implies. We then ask when this matching uniquely determines
the outcome change and describe estimation and predictive
adjustment. A Gaussian example illustrates what the stability
assumption rules out.

\subsection{Prediction data and the stable measurement model}

For each observation, let $Y$ denote the outcome of interest,
$Z$ an observed proxy for $Y$, and $C$ the observed context.
In the measurement applications that motivate this paper,
$Y$ is a reference outcome and $Z$ contains readings from
separate sensors. Prediction uses $X=(C,Z)$ to construct an
interval for $Y$. The source data contain both the inputs
and the outcome, whereas the target data used for adaptation
contain only the inputs.

Let $P$ and $Q$ denote the source and target distributions of
$(X,Y)$. We assume $Q\ll P$, so that the target-to-source
density ratio exists. We use mutually independent source
training, source calibration, and unlabeled target samples,
\begin{equation}
 \calD_{\rm tr}\sim P^N,\qquad
 \calD_{\rm cal}=\{(X_i,Y_i)\}_{i=1}^n\sim P^n,\qquad
 \calD_{\rm tar}=\{X_j^{\rm t}\}_{j=1}^m\sim Q_X^m.
 \label{eq:data}
\end{equation}
The superscript $\mathrm t$ marks target observations.
A new test pair $(X_*,Y_*)\sim Q$ is independent of all three samples.

Let $\calA=(\calD_{\rm tr},\calD_{\rm tar})$ denote the source
training and unlabeled target data.
Using only these samples, we construct predictive scores,
estimate density ratios, and choose among prespecified model
families. The resulting scores, ratios, and model choices
are fixed before calibration.
We seek marginal target coverage, with probability taken
over both the new target pair and the samples used to
construct its prediction set.
We describe the allowed changes between source and target
populations through the following mechanism-shift model.

\bigskip
\begin{definition}[Mechanism shift \citep{SubbaswamyA2022jci}]
\label{def:mechanismshift}
Let $\mathcal G$ be a causal directed acyclic graph on
$V=(V_1,\ldots,V_d)$, with parent sets $\pa_j$.
Suppose the source distribution factorizes as
\[
 p_V(v)=\prod_{j=1}^d p_j(v_j\mid v_{\pa_j}).
\]
For a specified set $\calJ\subseteq\{1,\ldots,d\}$, a target
distribution satisfies the mechanism-shift model restricted
to $\calJ$ if its density has the form
\[
 q_V(v)
 =\prod_{j\in\calJ}
   \underbrace{q_j(v_j\mid v_{\pa_j})}_{\text{may change}}
  \prod_{j\notin\calJ}
   \underbrace{p_j(v_j\mid v_{\pa_j})}_{\text{shared}}.
\]
Thus, mechanisms indexed by $\calJ$ may change, while all
other mechanisms remain the same in both populations.
For a root node, the corresponding mechanism is its
marginal distribution.
\end{definition}

We now apply this definition to the stable measurement model
in \Cref{fig:shiftoutcome}. In this graph, context $C$
influences both the outcome $Y$ and the proxy $Z$, and $Y$
influences $Z$. We specify the graph and the allowed changes
before fitting. The context distribution and the outcome
distribution within each context may change, while the
conditional distribution of the proxy given outcome and
context remains fixed. The source and target distributions
therefore have the form
\begin{align}
 p(c,y,z)
 &=p_C(c)p(y\mid c)k(z\mid y,c),\notag\\
 q(c,y,z)
 &=q_C(c)q(y\mid c)k(z\mid y,c).
 \label{eq:proxyfactor}
\end{align}
Here $k(z\mid y,c)$ is the \emph{conditional measurement density},
describing the distribution of possible readings at fixed outcome
and context, including measurement noise. This density represents
the measurement mechanism and is shared across populations.
The formulation also covers general proxies without physical sensors.

The causal interpretation treats
the reference outcome as the physical quantity generating the
sensor response. If that reference outcome is measured with
error, stability of the recorded $Z\mid Y,C$ remains an
assumption and does not follow from stable instrument
operation alone.

The conditional stability assumption in
\eqref{eq:proxyfactor} is the SJS restriction of
\citet[Definition~1]{ChenL2022neurips}, with $X_I=C$ and
$X_{I^c}=Z$. The joint distribution of $(C,Y)$ may change,
while the conditional distribution of $Z$ given $(C,Y)$ is
preserved. The shared measurement density cancels from the joint density
ratio, giving
\begin{equation}
 w((c,z),y)
 =\frac{q(c,y,z)}{p(c,y,z)}
 =\frac{q_C(c)}{p_C(c)}
  \frac{q(y\mid c)}{p(y\mid c)}.
 \label{eq:proxyw}
\end{equation}
The first factor accounts for changes in context frequencies.
The second accounts for outcome changes within context.
The dependence on $(C,Y)$ alone is the density-ratio
characterization of SJS given for finite outcome spaces by
\citet[Theorem~4.1]{TascheD2024fods}.
The same cancellation
applies here to continuous outcome densities.
Although $Z$ does not appear in the ratio, the unlabeled
target observations of $(C,Z)$ are used to estimate the
outcome correction.

\Cref{fig:shiftoverview} compares this restriction with
covariate shift and label shift. Panel~(a) keeps $Y\mid X$
fixed while allowing the input distribution to change.
Panel~(b) keeps $X\mid Y$ fixed while allowing the outcome
distribution to change. Panel~(c) represents our model,
in which both the input distribution and the predictive
distribution $Y\mid X$ may change.

Familiar special cases follow from \eqref{eq:proxyfactor}.
If only the context distribution changes, then $Y\mid C,Z$
remains fixed, giving covariate shift on $X=(C,Z)$.
With no context variable, the model reduces to label shift
with $X=Z$, the $0$-SJS case of
\citet[Theorem~2]{ChenL2022neurips}.
Label shift on the full input $(C,Z)$ also preserves
$Z\mid Y,C$. This inclusion follows from the nesting property
of SJS \citep[Corollary~4.3]{TascheD2024fods}.
For a fixed context--proxy split, however, the model does
not include every covariate shift, since covariate shift
alone need not preserve $Z\mid Y,C$.
We next use the shared measurement mechanism to match the observed target
input distribution, then ask whether this matching determines
the correction in \eqref{eq:proxyw}.

\begin{figure}[!t]
\centering
\begin{subfigure}[t]{.48\textwidth}
\centering
\begin{tikzpicture}[>=Latex]
 \path[use as bounding box]
   (-.25,-1.05) rectangle (5.65,2.2);
 \node[cmsenv] (e) at (1.1,1.65) {$E$};
 \node[obs,cmsnode,cmschange] (x) at (1.1,0) {$X$};
 \node[latent,cmsnode] (y) at (4.3,0) {$Y$};
 \cmsedge {x} {y}
 \draw[cmsshift] (e) -- (x);
 \draw[cmsinput] (.48,-.58) rectangle (1.72,.58);
\end{tikzpicture}
\par\smallskip
{\small May change $p_X$ to $q_X$.\\[2pt]
Stable $p(y\mid x)$.\par}
\caption{Covariate shift.}
\label{fig:shiftcovariate}
\end{subfigure}\hfill
\begin{subfigure}[t]{.48\textwidth}
\centering
\begin{tikzpicture}[>=Latex]
 \path[use as bounding box]
   (-.25,-1.05) rectangle (5.65,2.2);
 \node[cmsenv] (e) at (4.3,1.65) {$E$};
 \node[obs,cmsnode] (x) at (1.1,0) {$X$};
 \node[latent,cmsnode,cmschange] (y) at (4.3,0) {$Y$};
 \cmsedge {y} {x}
 \draw[cmsshift] (e) -- (y);
 \draw[cmsinput] (.48,-.58) rectangle (1.72,.58);
\end{tikzpicture}
\par\smallskip
{\small May change $p_Y$ to $q_Y$.\\[2pt]
Stable $p(x\mid y)$.\par}
\caption{Label shift.}
\label{fig:shiftlabel}
\end{subfigure}

\medskip
\begin{subfigure}[t]{.48\textwidth}
\centering
\begin{tikzpicture}[>=Latex]
 \path[use as bounding box]
   (-.25,-1.05) rectangle (5.65,2.2);
 \node[cmsenv] (e) at (.5,1.65) {$E$};
 \node[obs,cmsnode,cmschange] (c) at (.5,0) {$C$};
 \node[latent,cmsnode,cmschange] (y) at (2.7,1.65) {$Y$};
 \node[obs,cmsnode] (z) at (4.9,0) {$Z$};
 \cmsedge {c} {y,z}
 \cmsedge {y} {z}
 \draw[cmsshift] (e) -- (y);
 \draw[cmsshift] (e) -- (c);
 \draw[cmsinput] (-.12,-.58) rectangle (5.52,.58);
 \node[cmsnote] at (2.7,-.9) {$X=(C,Z)$};
\end{tikzpicture}
\par\smallskip
{\small May change $p_C$ and $p(y\mid c)$.\\[2pt]
Stable $k(z\mid y,c)$.\par}
\caption{Stable measurement mechanism with context.}
\label{fig:shiftoutcome}
\end{subfigure}

\caption{Three mechanism-shift models viewed during target
adaptation. The square $E$ denotes the environment, with
$E=0$ for source and $E=1$ for target. It is not a prediction
feature. Solid arrows specify the assumed causal structure.
Dashed arrows from $E$ and orange outlines mark mechanisms
allowed to change. All other mechanisms remain fixed.
Gray circles indicate variables observed during target
adaptation, while white circles indicate variables then
unavailable. Dashed blue boxes identify the prediction
input $X$. Panels (a) and (b) use different causal orientations.
Panel (c) represents the model in \eqref{eq:proxyfactor}.}
\label{fig:shiftoverview}
\end{figure}
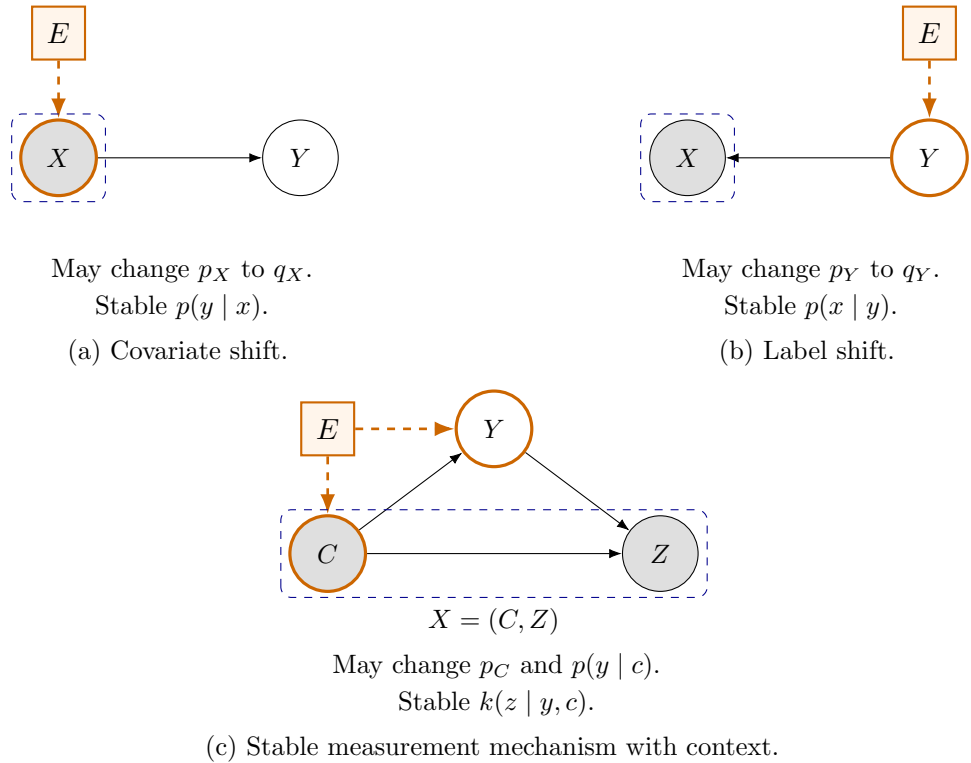

\subsection{Marginal matching and identification}
\label{subsec:identification}

Let $q(y\mid c)$ denote the unknown target outcome distribution.
We consider a prespecified family of candidates
$\{q_\beta(y\mid c):\beta\in\mathcal B\}$, indexed by a parameter
$\beta$. For now, treat the source distribution and the shared
measurement density $k$ as known. \Cref{subsec:fitratio}
specifies a conditional exponential-tilt family and explains
how to fit it when these source quantities are estimated.

\paragraph{From candidate outcomes to sensor readings.}
The source data contain $(C,Y,Z)$, allowing us to learn how
sensor readings depend on outcome and context. Under the
assumed stability of this relationship, the target sensor
distribution satisfies
\begin{equation}
 q(z\mid c)
 =\int k(z\mid y,c)q(y\mid c)\dd y.
 \label{eq:channelmatching}
\end{equation}
This integral averages the sensor response over the outcomes
that occur in the target population. Replacing the unknown
outcome distribution by a candidate gives the sensor and
input distributions
\begin{equation}
 \begin{split}
 q_\beta(z\mid c)
 &=\int k(z\mid y,c)q_\beta(y\mid c)\dd y,\\
 q_{\beta,X}(c,z)
 &=q_C(c)q_\beta(z\mid c).
 \end{split}
 \label{eq:impliedinput}
\end{equation}
The context distribution $Q_C$ is directly observed in the
target population. The term \emph{marginal matching} refers
to integrating out $Y$ and matching the resulting distribution
of the inputs $X=(C,Z)$.

\paragraph{Choosing the correction by matching inputs.}
At the population level, we choose $\beta$ to minimize
$\KL(Q_X\|Q_{\beta,X})$. Because the candidates share the
target context distribution, and assuming the log-density
expectations are finite,
\begin{align}
 \KL(Q_X\|Q_{\beta,X})
 &=\E_{C\sim Q_C}
   \KL(Q_{Z\mid C}\|Q_{\beta,Z\mid C})\notag\\
 &=\text{constant}
   -\E_{Q_X}\log q_\beta(Z\mid C).
 \label{eq:inputkl}
\end{align}
The first equality expresses the comparison within each
context, averaged over target contexts. The constant in the
second equality does not depend on $\beta$. Minimizing the
divergence therefore amounts to maximizing the expected
conditional log likelihood of the target sensor readings.
For observed target inputs $(c_j^{\rm t},z_j^{\rm t})$,
we replace this expectation by its sample average to obtain
\[
 \widehat\beta\in\argmax_{\beta\in\mathcal B}
 \frac{1}{m}\sum_{j=1}^m
 \log q_\beta(z_j^{\rm t}\mid c_j^{\rm t}).
\]
It requires neither target outcomes nor a separate estimate
of the target sensor density.

For a simple example, fix a context and suppose the shared
measurement model is $Z=2Y+\varepsilon$, where
$\varepsilon\sim\N(0,1)$ is independent of $Y$.
A candidate outcome distribution $Y\sim\N(\mu,1)$ implies
$Z\sim\N(2\mu,5)$. Matching a target sensor distribution
$\N(2,5)$ therefore selects $\mu=1$. With a finite target
sample, maximum likelihood gives
$\widehat\mu=\overline Z_{\rm t}/2$, where
$\overline Z_{\rm t}$ is the target sample mean.
The outcome mean is inferred through its effect on the
observed readings.

Estimating an outcome correction by matching its implied input distribution 
to the target input distribution is an established principle. 
It has been used in label-shift and target-shift estimation \citep{SaerensM2002neco,duPlessisM2012icml,ZhangK2013icml}, 
sparse joint shift \citep{ChenL2022neurips, TascheD2024fods}, 
and structured exponential tilting \citep{MaityS2023iclr, Choi2026arxiv_extracp}.
Here we adopt this principle within the stable measurement model and study the consequences of 
the estimated correction for conformal coverage and interval length.

\paragraph{When matching identifies the outcome distribution.}
If the true outcome distribution belongs to the candidate family,
it gives zero input divergence. But another outcome distribution
might produce exactly the same sensor distribution and therefore
the same fit. To recover the outcome correction uniquely, we need
distinct admissible outcome distributions to produce distinct
sensor distributions. This is the injectivity requirement.

For a candidate outcome distribution $\nu$, define
\[
 (K_c\nu)(B)
 =\int k(B\mid y,c)\nu(\dd y),
\]
where $k(B\mid y,c)$ is the probability that the sensor reading
belongs to a measurable set $B$. Thus, $K_c$ maps an outcome
distribution to the sensor distribution it produces.
The following lemma is a direct application of the standard
injectivity argument within each context. The same uniqueness
requirement appears in target-shift estimation
\citep[Section~3.1, Assumption~$A_3^{\rm TarS}$]{ZhangK2013icml}.
For finite outcome spaces, related sparse-joint-shift results
use linear-independence or full-rank conditions
\citep[Theorem~1]{ChenL2022neurips}
\citep[Theorem~4.7]{TascheD2024fods}.
We use this established principle to state exactly what the
assumed measurement mechanism identifies in our model. The short proof is
included for completeness.

\bigskip
\begin{lemma}[Identification through a stable measurement mechanism]
\label{thm:proxyid}
Fix the source distribution $P$ and the shared measurement density $k$,
and consider target distributions satisfying
\eqref{eq:proxyfactor} and $Q\ll P$.
Suppose that, for $Q_C$-almost every $c$,
\[
 K_c\nu_1=K_c\nu_2
 \quad\Longrightarrow\quad
 \nu_1=\nu_2
\]
for all admissible outcome distributions $\nu_1$ and $\nu_2$.
Then the target input distribution $Q_{C,Z}$ uniquely
determines the joint target distribution $Q$.
Any two density ratios compatible with this input
distribution and the assumed model agree with probability
one under $P$.
\end{lemma}

\begin{proof}
Any two compatible target distributions have the same
context distribution and the same conditional sensor
distribution. Injectivity therefore forces their conditional
outcome distributions to agree for $Q_C$-almost every context.
Substitution into \eqref{eq:proxyfactor} gives the same joint
target distribution. Since $Q\ll P$, its density ratio
relative to $P$ is unique with probability one under $P$.
\end{proof}

The lemma concerns identification at the population level.
If the family is misspecified, minimizing KL divergence
instead seeks the closest input distribution that the
model can produce. A good input fit then need not imply
an accurate outcome correction.

The identification requirement can also be expressed through
density ratios. For any $Q\ll P$,
\begin{equation}
 r_X(x)
 =\frac{q_X(x)}{p_X(x)}
 =\E_P[w(X,Y)\mid X=x].
 \label{eq:inputratio}
\end{equation}
The observed input ratio is therefore an average of the
joint ratio over the unobserved outcome. To identify the
joint ratio through marginal matching, distinct candidate
joint ratios must produce distinct input ratios.
A classifier distinguishing source from target inputs can
estimate $r_X$ through its class odds, adjusted for the
source and target sampling proportions.
Under covariate shift, $w(x,y)=r_X(x)$, so this input ratio
provides the required correction.
Under the stable measurement model, $w$ generally also
depends on $y$. Recovering that dependence requires the
measurement stability assumption and the identification condition.

Identification alone does not guarantee accurate estimation
from a finite sample. A known additive Gaussian measurement model
with nonzero loading gives an injective map from outcome to sensor
distributions. However, when the loading is small relative to the
noise, different outcome distributions
can produce very similar sensor distributions. A finite
target sample may then provide little information for
distinguishing them. Estimating the source measurement density adds
further uncertainty. These estimation difficulties arise
even when the stability assumption is correct. Whether
the measurement mechanism actually remains stable is a separate issue.
Matching the target input distribution alone cannot
establish that assumption.

\subsection{Estimation and predictive adjustment}
\label{subsec:fitratio}
\label{sec:proxy}

We now implement marginal matching with a conditional
exponential-tilt family.

\paragraph{Model for outcome changes.}
We specify the candidate outcome distributions using a
conditional exponential tilt of the source distribution,
\begin{equation}
 \begin{split}
 q_\beta(y\mid c)
 &=p(y\mid c)
   \exp\left\{\beta^\top t(y,c)-A(\beta,c)\right\},\\
 A(\beta,c)
 &=\log\E_P\!\left[
   e^{\beta^\top t(Y,c)}\mid C=c
   \right].
 \end{split}
 \label{eq:localtilt}
\end{equation}
Here $t(y,c)$ is a prespecified vector of functions of the
outcome and context. The normalizer $A(\beta,c)$ ensures
that $q_\beta(y\mid c)$ integrates to one within each
context. We restrict $\beta$ to values for which this
normalizer is finite in every relevant context.

For example, suppose $C$ is scalar and the source
conditional distribution is Gaussian, with variance
$v_P(c)$. Choosing $t(y,c)=(y,cy)^\top$ changes the
conditional mean by $v_P(c)(\beta_1+\beta_2c)$ while
leaving the conditional variance unchanged. Adding
a $y^2$ term also allows the conditional variance to
change, subject to the finite-normalizer requirement.

Normalization within each context allows us to specify
outcome changes separately from changes in context
frequencies. Under the stable measurement mechanism,
the resulting joint density ratio is
\[
 w_\beta((c,z),y)
 =\frac{q_C(c)}{p_C(c)}
  \exp\left\{\beta^\top t(y,c)-A(\beta,c)\right\}.
\]
Choosing $q_C=p_C$
preserves the context distribution while allowing
$Y\mid C$ to change.

\paragraph{Estimation from observed samples.}
We first use the source training data $\calD_{\rm tr}$ to
estimate the outcome distribution and conditional measurement density,
denoting the estimates by $p_{\widehat\theta}(y\mid c)$ and
$k_{\widehat\theta}(z\mid y,c)$. For each candidate $\beta$,
we construct $q_{\widehat\theta,\beta}(y\mid c)$ by substituting
$p_{\widehat\theta}(y\mid c)$ into \eqref{eq:localtilt} and
computing the normalizer under this same fitted distribution.

We hold these source estimates fixed while fitting $\beta$
from the unlabeled target inputs. Substituting the source
estimates into \eqref{eq:inputkl} and replacing the target
expectation by its sample average gives
\begin{equation}
 \widehat\beta\in\argmax_\beta
 \frac{1}{m}\sum_{j=1}^m
 \log\!\int
 q_{\widehat\theta,\beta}(y\mid c_j^{\rm t})
 k_{\widehat\theta}(z_j^{\rm t}\mid y,c_j^{\rm t})
 \dd y.
 \label{eq:proxylikelihood}
\end{equation}
Constraints or penalties may be used to regularize the fit.
We estimate the context ratio $r_C=q_C/p_C$ separately from
the source training and target context observations.
Combining the context and outcome corrections gives the
calibration weight
\begin{equation}
 \widehat w((c,z),y)
 =\widehat r_C(c)
  \frac{q_{\widehat\theta,\widehat\beta}(y\mid c)}
       {p_{\widehat\theta}(y\mid c)}.
 \label{eq:fittedproxyweight}
\end{equation}
All estimates and tuning choices use only $\calA$ and are
fixed before calibration. The accuracy of the resulting
weights depends on the source estimates, the choice of
outcome family, and the information about outcome changes
contained in the target sensor readings.

\paragraph{Predictive adjustment.}
The density ratio can also adapt the source predictive
distribution to the target population. At a fixed input,
we multiply the source predictive density by the ratio
and normalize over possible outcomes. Let $w_\beta$ denote
the ratio in \eqref{eq:proxyw} with $q(y\mid c)$ replaced
by $q_\beta(y\mid c)$. The corresponding predictive
distribution is
\begin{equation}
 q_\beta(y\mid x)
 =\frac{p(y\mid x)w_\beta(x,y)}
        {\int p(u\mid x)w_\beta(x,u)\dd u}.
 \label{eq:predictive}
\end{equation}
This is the standard conditional density-reweighting
identity, also used for predictive correction under SJS
with finite outcome classes \citep[Proposition~4.5]{TascheD2024fods}.
In practice, we substitute the fitted source predictive
distribution and $\widehat w$.
Factors depending only on the input, including the context
ratio, cancel because the input is fixed in this conditional
distribution. They remain necessary in the calibration
weights, which must account for differences in how often
inputs occur in the source and target populations.

To examine the effect of predictive adjustment, we compare
a source score (S) with an adjusted score (T), using the
same fitted weights and calibration sample.
For predictive mean $\mu_{\rm pred}(x)$ and variance
$v_{\rm pred}(x)>0$, a standard choice is
\begin{equation}
 s(x,y)
 =\frac{|y-\mu_{\rm pred}(x)|}{\sqrt{v_{\rm pred}(x)}}.
 \label{eq:stdscore}
\end{equation}
The source score uses the source predictive moments,
whereas the adjusted score uses the target-adjusted
predictive moments. We also use unstandardized absolute
residuals in the main Gaussian comparison.
With the true density ratio and exact weighted conformal
calibration, either score gives marginal target coverage
when fixed before calibration using only $\calA$.
Changing the score can affect the size of the prediction
sets, but shorter sets are not guaranteed.
\Cref{sec:calibration} describes the effects of estimated
weights and the finite weighted-quantile approximation.
Algorithm~\ref{alg:stablemeasurement} summarizes the full
procedure, from fitting the correction to producing a prediction set.

\subsection{What the measurement restriction rules out}
\label{subsec:jointambiguity}

Within a given context, the distribution of sensor readings can change because the outcome distribution changes, 
because the measurement process changes, or both. 
The following example compares the stable measurement model with covariate shift. 
Both can reproduce the same target input distribution while implying different outcome predictions. 
Thus, matching the target inputs alone cannot determine which assumption is appropriate.

Suppose the source model is
\[
 Y\mid C=c\sim\N(0,1),\qquad
 Z=\lambda Y+\varepsilon,\qquad
 \varepsilon\sim\N(0,\tau^2),
\]
where $\varepsilon$ is independent of $(C,Y)$,
$\lambda\ne0$, and $\tau^2>0$.
Write $\sigma_Z^2=\lambda^2+\tau^2$ for the source variance of $Z$
within each context. To allow the conditional distribution
of $Z$ given $Y$ and $C$ to change as well, consider a
joint-tilt family with separate coefficients for $Y$ and $Z$,
\begin{equation}
 \begin{split}
 q_{\beta_Y,\beta_Z}(y,z\mid c)
 &=p(y,z\mid c)
   \exp\{\beta_{Y,c}y+\beta_{Z,c}z-A_c^{\rm joint}\},\\
 A_c^{\rm joint}
 &=\tfrac12\bigl(
   \beta_{Y,c}^2+2\lambda\beta_{Y,c}\beta_{Z,c}
   +\sigma_Z^2\beta_{Z,c}^2
   \bigr).
 \end{split}
 \label{eq:jointfamily}
\end{equation}
A linear exponential tilt of a Gaussian distribution preserves
its covariance matrix and shifts each mean according to the
covariances with the tilted variables. Here the source variance of $Y$ is $1$,
the covariance between $Y$ and $Z$ is $\lambda$, and the variance
of $Z$ is $\sigma_Z^2$. The target means are therefore
\[
 \E_Q(Y\mid C=c)=\beta_{Y,c}+\lambda\beta_{Z,c},
 \qquad
 \E_Q(Z\mid C=c)=\lambda\beta_{Y,c}+\sigma_Z^2\beta_{Z,c}.
\]

Fix the target context distribution and suppose that
$Z\mid C=c\sim\N(\mu_{Z,c},\sigma_Z^2)$ in the target population.
Every coefficient pair satisfying
\[
 \lambda\beta_{Y,c}+\sigma_Z^2\beta_{Z,c}=\mu_{Z,c}
\]
reproduces this sensor distribution. This is one equation
for two tilt coefficients. Different pairs on this line give
different outcome means, so input matching alone does not
identify the outcome correction in this family.

The two shift assumptions resolve the ambiguity by fixing
different conditional distributions. Under \eqref{eq:jointfamily},
\[
 \begin{aligned}
 Z\mid Y=y,C=c
 &\sim\N(\lambda y+\tau^2\beta_{Z,c},\tau^2),\\
 Y\mid Z=z,C=c
 &\sim\N\!\left(
   \frac{\lambda z}{\sigma_Z^2}+\frac{\tau^2\beta_{Y,c}}{\sigma_Z^2},
   \frac{\tau^2}{\sigma_Z^2}
   \right).
 \end{aligned}
\]
Stable measurement keeps the first conditional distribution
equal to its source counterpart. This requires
$\beta_{Z,c}=0$, after which input matching gives
$\beta_{Y,c}=\mu_{Z,c}/\lambda$.
Covariate shift keeps the second conditional distribution
unchanged. This requires $\beta_{Y,c}=0$, after which input
matching gives $\beta_{Z,c}=\mu_{Z,c}/\sigma_Z^2$.
Consequently,
\[
 \begin{array}{ll}
 \text{Stable measurement model}
 &\E_Q(Y\mid C=c)=\mu_{Z,c}/\lambda,\\[2pt]
 \text{Covariate shift}
 &\E_Q(Y\mid C=c)=\lambda\mu_{Z,c}/\sigma_Z^2.
 \end{array}
\]
These means differ whenever $\mu_{Z,c}\ne0$, despite the
identical target input distributions.

For example, take $\lambda=1$, $\tau^2=1$, and
$\mu_{Z,c}=2$. Both models give the target sensor distribution
$\N(2,2)$, but the stable measurement model gives an outcome
mean of $2$, whereas covariate shift gives $1$.
At the same observed reading $z$, the predictive mean is
$z/2+1$ under stable measurement and $z/2$ under covariate shift.

The example separates identification within a model from
the choice of that model. Either restriction selects a unique
correction here, but matching the target inputs cannot distinguish
between them. In particular, a perfect match does not establish
measurement stability. The ambiguity arises from allowing both
tilts in \eqref{eq:jointfamily}. A more restricted joint-tilt
family can be identifiable.

\section{Conformal calibration, context, and estimation error}
\label{sec:gaussian}

This section separates the benefit of using context from
the cost of estimating the target correction. We first recall
the weighted-conformal guarantee and its fitted-ratio bound.
We then show how context can shorten intervals when the target
distribution is known. A Gaussian calculation evaluates the
coverage lost when the correction must be estimated, and a
comparison of shared and separate corrections explains how
to retain context while limiting that estimation cost.

\subsection{Weighted calibration with oracle and estimated ratios}
\label{sec:calibration}

The density ratio gives greater weight to source observations
that are more prevalent in the target population. Weighted
conformal calibration uses these weights to compare a candidate
outcome with the source calibration observations. With the true
ratio, its guarantee follows from weighted exchangeability
\citep{TibshiraniR2019neurips}. Estimating the ratio introduces
an additional error, which can be bounded by the discrepancy
between the true target distribution and the distribution
represented by the fitted weights \citep{LeiL2021jrsssb}.

Fix $0<\alpha<1$ and condition on the source training and
unlabeled target data $\calA$. The score is then fixed, while
the calibration observations and the new target pair remain
independent of $\calA$.
Write $S_i=s(X_i,Y_i)$ and $w_i=w(X_i,Y_i)$ for
$i=1,\ldots,n$. Larger scores indicate poorer agreement with
the predictor. For a candidate pair $(x,y)$, define
\begin{equation}
 \begin{split}
 p_w(x,y)
 &=\frac{w(x,y)+\sum_{i=1}^n w_i\ind\{S_i\ge s(x,y)\}}
 {w(x,y)+\sum_{i=1}^n w_i},\\
 \calC_w(x)&=\{y:p_w(x,y)>\alpha\}.
 \end{split}
 \label{eq:wpvalue}
\end{equation}
Set $p_w(x,y)=1$ if the denominator is zero.
The numerator includes the candidate's own weight and the
weights of calibration observations with scores at least as
large. A candidate is retained when this weighted upper-tail
probability exceeds $\alpha$.

For the true joint ratio $w=\dd Q/\dd P$,
\[
 \Prob\{Y_*\in\calC_w(X_*)\mid\calA\}\ge1-\alpha.
\]
This is a direct application of the weighted-exchangeability
result of \citet{TibshiraniR2019neurips} to the prediction pair
$(X,Y)$. Conditional on the unordered collection of all $n+1$
observations, the probability that an observation is the target
draw is proportional to its weight. The rank rule excludes
indices with total probability at most $\alpha$, including
when scores are tied. The guarantee holds for any score
constructed from $\calA$.

When $w(x,y)$ depends on the candidate outcome, the exact set
need not be a bounded interval. A sufficiently large candidate
weight can retain even outcomes with large scores. For the
interval-length comparisons, we therefore use the finite
weighted-quantile (FWQ) rule
\begin{equation}
 \begin{split}
 \widehat q_w
 &=\inf\left\{s_0:
 \frac{\sum_{i=1}^n w_i\ind\{S_i\le s_0\}}
      {\sum_{i=1}^n w_i}\ge1-\alpha\right\},\\
 \calC_{\rm fwq}(x)&=\{y:s(x,y)\le\widehat q_w\}.
 \end{split}
 \label{eq:fwq}
\end{equation}
If the calibration weights sum to zero, return $\calY$.
Otherwise, this rule uses only the calibration weights and
produces a subset of $\calC_w(x)$. For the standardized
residual score in \eqref{eq:stdscore}, it gives
$\mu_{\rm pred}(x)\pm\widehat q_w\sqrt{v_{\rm pred}(x)}$.
The following lemma bounds the coverage reduction from
omitting the candidate weight.

The lemma is a supporting consequence of the weighted-index
argument and the largest-weight quantile bound in
\citet{TibshiraniR2019neurips,LeiL2021jrsssb}, adapted here
to the FWQ rule. An elementary proof allowing tied scores
is given in \Cref{app:fwqproof}.

\bigskip
\begin{lemma}[Finite-quantile discretization bound]
\label{prop:fwq}
Under \eqref{eq:data}, use the true ratio and a score fixed
by $\calA$. Put $w_{n+1}=w(X_*,Y_*)$. Then
\begin{equation}
 \Prob\{Y_*\notin\calC_{\rm fwq}(X_*)\mid\calA\}
 \le\alpha+\E[\pi_{\max}\mid\calA],
 \qquad
 \pi_{\max}
 =\max_{1\le i\le n+1}\frac{w_i}{\sum_{\ell=1}^{n+1}w_\ell}.
 \label{eq:fwqcoverage}
\end{equation}
\end{lemma}
The extra term is small when no single observation carries a
large fraction of the total weight. It accounts for the finite
quantile approximation even when the ratio is known exactly.

To describe the effect of estimating the ratio, let
$\widehat w\ge0$ be fixed by $\calA$ and suppose
$0<\E_P\widehat w<\infty$. Reweighting the true source
distribution by $\widehat w$ defines
\begin{equation}
 \frac{\dd R_{\widehat w}}{\dd P}=\overline w,
 \qquad
 \overline w=\frac{\widehat w}{\E_P\widehat w}.
 \label{eq:effective}
\end{equation}
Here the expectation is over a fresh source observation, with
the fitted function held fixed. The distribution $R_{\widehat w}$
is the target distribution represented by the estimated weights.
The normalizing constant is needed to define this distribution,
but cancels from the calibration rules.

Exact weighted ranks with $\widehat w$ have coverage at least
$1-\alpha$ under $R_{\widehat w}$. Changing the test distribution
from $R_{\widehat w}$ to $Q$ changes the coverage probability
by at most
\[
 \TV(Q,R_{\widehat w})
 =\tfrac12\E_P|w-\overline w|.
\]
This is the total-variation transfer used for estimated
covariate weights by \citet{LeiL2021jrsssb}, applied here
to the joint distribution of inputs and outcomes.
Combining it with \Cref{prop:fwq} gives the following
immediate corollary.

\bigskip
\begin{corollary}[Fitted finite quantiles]
\label{cor:fwqfitted}
Under \eqref{eq:data}, let the score and a nonnegative fitted
weight $\widehat w$ be fixed by $\calA$, with
$0<\E_P\widehat w<\infty$. The FWQ set constructed with
$\widehat w$ satisfies
\begin{equation}
 \Prob\{Y_*\in\calC_{\rm fwq}(X_*)\mid\calA\}
 \ge1-\alpha
 -\E_{P^n\otimes R_{\widehat w}}
    [\widehat\pi_{\max}\mid\calA]
 -\TV(Q,R_{\widehat w}),
 \label{eq:fittedfwq}
\end{equation}
where $\widehat\pi_{\max}$ is the largest normalized weight
among the $n$ source calibration pairs and one independent
pair drawn from $R_{\widehat w}$, all weighted by $\widehat w$.
The result also holds when the model family or score is
selected using $\calA$ alone.
\end{corollary}

The bound separates error from omitting the candidate weight
from error in estimating the target distribution. For exact ranks,
only the total-variation term remains. In particular, if
$\widehat w=\widehat q/\widehat p$, then
$R_{\widehat w}\propto P\widehat q/\widehat p$.
It need not equal $\widehat Q$, because calibration observations
come from the true source distribution $P$. Errors in the
source model therefore also affect the fitted-weight guarantee.

Algorithm~\ref{alg:stablemeasurement} collects the estimation
steps from \Cref{subsec:fitratio} and the two calibration rules.
The displayed score uses
standardized residuals. For an unstandardized residual score,
omit division by $\sqrt{v_{\rm pred}(x)}$ and the corresponding factor
in the returned FWQ interval.

\begin{algorithm}[ht!]
\caption{Stable measurement adjustment and weighted conformal prediction}
\label{alg:stablemeasurement}
\small
\begin{algorithmic}[1]
\Require Independent samples $\calD_{\rm tr},\calD_{\rm tar},\calD_{\rm cal}$,
level $\alpha$, test input $x_*=(c_*,z_*)$, a prespecified tilt family,
score S or T, and FWQ or exact-rank calibration.
\Ensure A prediction set $\widehat{\calC}(x_*)$.
\Statex \textit{Fit using only $\calA=(\calD_{\rm tr},\calD_{\rm tar})$.}
\State Fit $p_{\widehat\theta}(y\mid c)$ and
$k_{\widehat\theta}(z\mid y,c)$ on $\calD_{\rm tr}$ and form
$p_{\widehat\theta}(y\mid c,z)\propto
p_{\widehat\theta}(y\mid c)k_{\widehat\theta}(z\mid y,c)$,
normalized over $y$.
\State Estimate $\widehat r_C(c)$ from source training and target contexts.
\State Fit $\widehat\beta$ by the target-input likelihood
\eqref{eq:proxylikelihood}, using the fitted-source normalizer in
\eqref{eq:localtilt}. Choose any constraints or penalties using only $\calA$.
\State Set $\widehat w((c,z),y)\gets\widehat r_C(c)
q_{\widehat\theta,\widehat\beta}(y\mid c)/p_{\widehat\theta}(y\mid c)$.
\State For score S, use the source predictive distribution.
For T, reweight it by \eqref{eq:predictive} with $\widehat w$.
\State Compute the chosen predictive mean $\mu_{\rm pred}(x)$ and variance $v_{\rm pred}(x)>0$
and set $s(x,y)\gets|y-\mu_{\rm pred}(x)|/\sqrt{v_{\rm pred}(x)}$.
\State Freeze the score, weights, and all model choices before calibration.
\Statex \textit{Calibrate using $\calD_{\rm cal}$.}
\State Compute $S_i=s(X_i,Y_i)$ and $\widehat w_i=\widehat w(X_i,Y_i)$
for $i=1,\ldots,n$, and set $W\gets\sum_{i=1}^n\widehat w_i$.
\If{exact-rank calibration is selected}
 \State Evaluate $\widehat p(x_*,y)$ by \eqref{eq:wpvalue} with
 $\widehat w$, retaining the candidate weight $\widehat w(x_*,y)$.
 Set $\widehat p(x_*,y)=1$ if the denominator is zero.
 \State \Return $\{y\in\calY:\widehat p(x_*,y)>\alpha\}$.
\Else \Comment{FWQ calibration used in the experiments}
 \If{$W=0$}
  \State \Return $\calY$.
 \EndIf
 \State Compute the weighted $(1-\alpha)$ quantile
 $\widehat q_{\widehat w}$ by \eqref{eq:fwq}, using only calibration weights.
 \State \Return $\{y\in\calY:|y-\mu_{\rm pred}(x_*)|
 \le\widehat q_{\widehat w}\sqrt{v_{\rm pred}(x_*)}\}$.
\EndIf
\end{algorithmic}
\end{algorithm}

\subsection{Prediction benefit of context}
\label{subsec:contextvalue}

The same sensor reading can imply different outcomes in
different contexts. For example, identical readings under
different weather conditions may correspond to different
pollutant concentrations. Using context allows a prediction
interval to account for these differences. We compare intervals
that use both $(C,Z)$ with intervals that use $Z$ alone,
assuming that the target distribution is known. This separates
the predictive benefit of context from the cost of estimating
the outcome correction.

Consider a Gaussian conditional outcome distribution with
the same variance at every input. Intervals centered at the
conditional mean can then achieve the desired coverage with
the same width for every observation. The next proposition
shows that they minimize expected interval length and that
ignoring informative context requires longer intervals on
average. This is a Gaussian consequence of the classical
minimum-volume-set principle \citep{ScottC2006jmlr}.
Write $\Phi$ and $\phi$ for the standard normal distribution
function and density.

\bigskip
\begin{proposition}[Context and minimum interval length]
\label{prop:contextvalue}
Suppose that, under $Q$,
\[
 Y\mid C,Z\sim\N(\mu_{\rm pred}(C,Z),\sigma^2)
\]
with constant $\sigma>0$. Every interval $I(C,Z)$ with finite
expected length and marginal coverage at least $1-\alpha$
satisfies
\begin{equation}
 \E_Q|I(C,Z)|\ge2\sigma z_\alpha,
 \qquad z_\alpha=\Phi^{-1}(1-\alpha/2).
 \label{eq:contextlengthbound}
\end{equation}
Equality is attained by
$\mu_{\rm pred}(C,Z)\pm\sigma z_\alpha$.
If $\Var_Q\{\mu_{\rm pred}(C,Z)\mid Z\}>0$ on a set of
positive probability, every interval depending only on $Z$
and meeting the same coverage requirement has strictly
greater expected length.
\end{proposition}

The conditional-variance condition means that context helps
determine the outcome mean even after the sensor reading is
known. In this case, an interval that ignores context cannot
attain the minimum expected length available to an interval
that uses it. The proof in \Cref{app:contextproof} combines
the optimal centering of a normal interval with Jensen's
inequality and its equality conditions.

\subsection{Coverage with oracle and fitted Gaussian corrections}
\label{subsec:scalarlearning}

We compare an oracle procedure that uses the true target
correction with a procedure that estimates the correction from
$m$ unlabeled target readings. Both are evaluated under the
same true target distribution $Q$. We assume that the source
distribution and measurement mechanism are known, and let the
independent source calibration size $n$ tend to infinity.
This isolates coverage error from estimating the correction.
Our aim is to quantify the resulting coverage gap relative
to the oracle and the additional cost of estimating the variance.

\paragraph{Gaussian model and estimation.}
Consider one outcome and one sensor, with no context variable.
This is the label-shift special case of the stable measurement
model. Under the source distribution $P$, let $Y\sim\N(0,1)$.
Under the target distribution $Q$, suppose
\begin{equation}
 Y\sim\N(\mu,v),\qquad
 Z=\lambda Y+\varepsilon,\qquad
 \varepsilon\sim\N(0,\tau^2),\qquad
 \sigma_Z^2=\lambda^2v+\tau^2,
 \label{eq:scalar}
\end{equation}
where $v>0$, $\lambda\ne0$, and $\tau^2>0$, with noise
independent of $Y$. The measurement mechanism is shared by
the two populations. We treat the source distribution and
the measurement parameters $(\lambda,\tau^2)$ as known.
Under $Q$, the sensor reading satisfies
$Z\sim\N(\lambda\mu,\sigma_Z^2)$, so its mean and variance
identify $(\mu,v)$.

Write $p_Y$ and $q_Y$ for the source and target outcome
densities, and $q(z\mid y)=p(z\mid y)$ for the shared
measurement density. Let $\phi_{\mu,v}$ denote the normal
density with mean $\mu$ and variance $v$. Then
\[
 p_Y(y)=\phi_{0,1}(y),\qquad
 q_Y(y)=\phi_{\mu,v}(y),\qquad
 q(z\mid y)=\phi_{\lambda y,\tau^2}(z).
\]

From $m\ge2$ independent target readings, estimate $(\mu,v)$ by
\begin{equation}
 \widehat\mu =\overline Z/\lambda,\qquad
 \widehat v
 =\left[\frac{S_Z^2-\tau^2}{\lambda^2}
   \right]_{[v_{\min},v_{\max}]}, 
 \label{eq:targetmle}
\end{equation}
where
\begin{equation*} 
\overline Z=m^{-1}\sum_{j=1}^m Z_j^{\rm t}, \qquad
 S_Z^2=\frac1m\sum_{j=1}^m(Z_j^{\rm t}-\overline Z)^2,
\end{equation*}
and the brackets clip the variance estimate to fixed bounds
$0<v_{\min}<v_{\max}<\infty$ chosen before fitting.
These are the constrained maximum-likelihood estimates.
The lower bound ensures a positive outcome variance even
when the sample sensor variance is no larger than the
known noise variance.

The fitted outcome density is
$\widehat q_Y(y)=\phi_{\widehat\mu,\widehat v}(y)$.
Because the measurement density is shared and known, the
fitted joint density ratio reduces to
\[
 \widehat w(y)
 =\frac{\widehat q_Y(y)}{p_Y(y)}
 =\frac{\phi_{\widehat\mu,\widehat v}(y)}{\phi_{0,1}(y)}.
\]
Applying Bayes' rule to the fitted outcome density and the
measurement likelihood gives
\[
 \begin{aligned}
 \widehat q(y\mid z)
 &\propto\widehat q_Y(y)q(z\mid y)\\
 &=\phi_{\widehat\mu,\widehat v}(y)\phi_{\lambda y,\tau^2}(z)\\
 &\propto\exp\left\{-\frac{(y-\widehat\mu)^2}{2\widehat v}
                    -\frac{(z-\lambda y)^2}{2\tau^2}\right\}.
 \end{aligned}
\]
Completing the square in $y$ gives the fitted conditional distribution
$Y\mid Z=z\sim\N(\widehat\mu_{\rm pred}(z),\widehat v_{\rm pred})$, with
\begin{equation}
 \widehat\mu_{\rm pred}(z)
 =\widehat\mu+\widehat\gamma(z-\lambda\widehat\mu),\qquad
 \widehat v_{\rm pred}=\frac{\widehat v\tau^2}{\lambda^2\widehat v+\tau^2},
 \label{eq:gaussianpredictive}
\end{equation}
where $\widehat\gamma
 =\frac{\lambda\widehat v}{\lambda^2\widehat v+\tau^2}$
is the regression coefficient of $Y$ on
$Z$ under the fitted model. It converts the sensor's deviation
from its expected reading, $z-\lambda\widehat\mu$, into an
adjustment to the estimated outcome mean $\widehat\mu$.
Estimating the outcome variance therefore affects the predictive
mean, the predictive variance, and the calibration weights.

\paragraph{Oracle and fitted prediction intervals.}
Let $\mu_{\rm pred}$ and $v_{\rm pred}$ denote the predictive
moments in \eqref{eq:gaussianpredictive} with
$(\widehat\mu,\widehat v)$ replaced by $(\mu,v)$.
The oracle uses the true ratio $w(y)=q_Y(y)/p_Y(y)$ and these
predictive moments. The fitted procedure uses $\widehat w$
and the fitted moments. Both use standardized residual scores.
Their large-calibration coverage is evaluated through the intervals
\[
 \begin{aligned}
 I_{\rm or}(z)&=\mu_{\rm pred}(z)\pm z_\alpha\sqrt{v_{\rm pred}},\\
 I_\infty(z)&=\widehat\mu_{\rm pred}(z)
                   \pm z_\alpha\sqrt{\widehat v_{\rm pred}}.
 \end{aligned}
\]

\paragraph{Coverage gap.}
The oracle interval has target coverage $1-\alpha$. Let
$C_m(\lambda,v)$ denote the coverage of $I_\infty$, averaged
over the target fitting sample and a new target pair under $Q$.
We compare the procedures through the signed coverage gap
\begin{equation}
 \Delta_m=(1-\alpha)-C_m(\lambda,v).
 \label{eq:coveragegap}
\end{equation}
Positive values indicate undercoverage, zero means nominal
coverage, and negative values indicate overcoverage.
\Cref{prop:scaleexact} characterizes this gap for fixed $m$.
\Cref{prop:scalecost} gives its leading behavior as $m$ grows.

\bigskip
\begin{proposition}[Exact coverage and the gap from the oracle]
\label{prop:scaleexact}
Under \eqref{eq:scalar}, fix $0<\alpha<1$ and $m\ge2$.
Use \eqref{eq:targetmle} for the fitted correction, with score
$\widehat s(z,y)=|y-\widehat\mu_{\rm pred}(z)|/\sqrt{\widehat v_{\rm pred}}$.
The target fitting sample, source calibration sample, and new
target pair are mutually independent, with independent
observations within each sample.
For either FWQ or exact weighted ranks, the following coverage
limits hold as $n\to\infty$ with $m$ fixed.

The oracle procedure has limiting target coverage
$\Prob_Q\{Y\in I_{\rm or}(Z)\}=1-\alpha$.
For the fitted procedure, coverage conditional on the fitted
parameters converges to
\begin{equation}
 \begin{split}
 C(\widehat\mu,\widehat v)
 &=\Prob_Q\{Y\in I_\infty(Z)\mid\widehat\mu,\widehat v\}\\
 &=\Phi\left(\frac{z_\alpha\sqrt{\widehat v_{\rm pred}}-b}{\sigma_{\rm res}}\right)
  -\Phi\left(\frac{-z_\alpha\sqrt{\widehat v_{\rm pred}}-b}{\sigma_{\rm res}}\right),\\
 b&=(\mu-\widehat\mu)(1-\lambda\widehat\gamma),\\
 \sigma_{\rm res}^2&=v(1-\lambda\widehat\gamma)^2+\widehat\gamma^2\tau^2,
 \end{split}
 \label{eq:conditionalcoverage}
\end{equation}
where $b$ and $\sigma_{\rm res}^2$ are the mean and variance of the prediction
residual under $Q$, given the fit. This conditional coverage
averages over a new target pair. Averaging also over the
target fitting sample gives the coverage and its gap from
the oracle,
\begin{equation}
 \begin{split}
 C_m(\lambda,v)
 &=\E_{\calD_{\rm tar}}\bigl[C(\widehat\mu,\widehat v)\bigr],\\
 \Delta_m
 &=\E_{\calD_{\rm tar}}\bigl[(1-\alpha)-C(\widehat\mu,\widehat v)\bigr].
 \end{split}
 \label{eq:scaleaverage}
\end{equation}
Both quantities are independent of $\mu$.
If the variance $v$ is known and only the mean $\mu$ is
estimated, use the superscript $\mathrm{loc}$ to denote
this \emph{location-only} estimation setting. The limiting
coverage, averaged over the target fitting sample, and its
gap from the oracle satisfy
\begin{equation}
 \begin{split}
 C_m^{\rm loc}
 &=2\Phi\left(\frac{z_\alpha}
 {\sqrt{1+\tau^2/(m\lambda^2v)}}\right)-1,\\
 \Delta_m^{\rm loc}
 &:=(1-\alpha)-C_m^{\rm loc}>0.
 \end{split}
 \label{eq:locationlimit}
\end{equation}
Thus mean estimation alone gives a strictly positive coverage
gap for every finite $m$, even as $n\to\infty$.
\end{proposition}

Weighted calibration with $\widehat w$ targets the fitted joint
distribution, under which $I_\infty$ has coverage $1-\alpha$.
The gap $\Delta_m$ compares this nominal coverage under the
fitted model with actual coverage under the true $Q$.
\Cref{app:exactproof} reduces the expectation in
\eqref{eq:scaleaverage} to a one-dimensional integral.
With both parameters estimated, clipping can give $\Delta_m<0$
at finite $m$ (\Cref{app:exactproof}). The next result establishes
positivity and compares the two gaps for sufficiently large $m$.

\paragraph{Coverage loss from estimating outcome parameters.}
Let $\kappa=\lambda^2v/\tau^2$ be the measurement signal-to-noise
ratio. Small $\kappa$ means that the readings contain little
information about the outcome. Estimating $v$ requires subtracting
the noise variance from the sample sensor variance and dividing
by $\lambda^2$, which amplifies errors when the loading is weak.
A second-order delta-method expansion
\citep{vanderVaartAW1998book} quantifies the resulting coverage loss.

\bigskip
\begin{proposition}[Leading coverage gaps from estimation]
\label{prop:scalecost}
Under the assumptions of \Cref{prop:scaleexact}, fix
$\kappa>0$, $0<\alpha<1$, and $v_{\min}<v<v_{\max}$.
After taking $n\to\infty$, let $m\to\infty$.
The mean-only and mean-and-variance coverage gaps satisfy
\begin{equation}
 \begin{split}
 \Delta_m^{\rm loc}
 &=\frac{z_\alpha\phi(z_\alpha)}{m\kappa}+o(m^{-1}),\\
 \Delta_m
 &=\frac{z_\alpha\phi(z_\alpha)}{m}
 \left\{\frac6\kappa+
              \frac{1+z_\alpha^2}{2\kappa^2}\right\}
 +o(m^{-1}).
 \end{split}
 \label{eq:scalecost}
\end{equation}
Consequently, $0<\Delta_m^{\rm loc}<\Delta_m$ for all sufficiently
large $m$, or equivalently $C_m<C_m^{\rm loc}<1-\alpha$.
\end{proposition}

The mean-only gap $\Delta_m^{\rm loc}$ has leading order
$1/(m\kappa)$. The difference $\Delta_m-\Delta_m^{\rm loc}$
measures the additional coverage loss from estimating the variance.
Subtracting the expansions gives a positive leading term,
including a contribution of order $1/(m\kappa^2)$. Weak measurements therefore make variance
estimation especially costly. The expansion holds at fixed
$\kappa$ with the true variance inside the bounds. When fitting
errors are large or clipping is frequent, use the exact expression
in \Cref{app:exactproof}.

For example, with $\alpha=0.10$, $\lambda=0.25$,
$v=\tau^2=1$, $m=2000$, and bounds $[0.05,10]$, the oracle
has coverage $0.90$, the mean-only fit has coverage $0.8986$,
and the mean-and-variance fit has coverage $0.8386$.
Thus $\Delta_m^{\rm loc}\approx0.0014$ and
$\Delta_m\approx0.0614$. The first-order expansion instead gives
coverage $0.8516$, understating the gap by about $21\%$. It is
not quantitatively accurate at this measurement strength and sample size.
The exact large-calibration values show that additional source
calibration data cannot remove these gaps.
At $\mu=0$, source and target distributions coincide, showing
that estimating an unnecessary correction can also be costly.
\Cref{sec:calibration} addresses finite calibration samples
and estimated source mechanisms. The present comparison isolates
target-fitting error with known source mechanisms.

\subsection{Shared versus separate corrections across contexts}
\label{subsec:contextcost}

Context is observed, but target outcomes are unlabeled, so the
outcome change within each context must be estimated from sensor
readings. We now consider the optional restriction that these
changes can be represented by one common correction. Measurement
stability alone does not impose this restriction. Both fitting
choices below retain context for prediction.

\paragraph{Two fitting choices.}
Write the source and target outcome means as
$\mu_c^P$ and $\mu_c^P+\delta_c$, respectively.
A separate fit estimates each $\delta_c$ using target readings
from that context. A shared fit assumes $\delta_c=\delta$ and
estimates one change using readings across contexts. For example,
a common increase of three units changes source means $10$ and
$20$ in cold and warm conditions to $13$ and $23$. Sharing the
change still allows different predictions in the two contexts.
Sharing can reduce fitting noise when target data are limited,
but introduces bias when the true changes differ.

\paragraph{Coverage trade-off.}
To isolate this choice, take the source mechanisms as known and
the outcome variances as known and unchanged. Let $n\to\infty$
before $m\to\infty$. For a direct comparison, fix $K$ and suppose
$v_c=v$ and $\kappa_c=\lambda_c^2v/\tau_c^2=\kappa>0$ in every
context, with target counts $m_c/m\to\pi_c>0$.
Consider $\delta_c=\delta_0+\xi_c/\sqrt m$, where
$\sum_c\pi_c\xi_c=0$. Differences between changes are then of
the same order as estimation error. Let $\Delta_m^{\rm separate}$
and $\Delta_m^{\rm shared}$ denote $1-\alpha$ minus the respective
limiting target coverages, averaged over the target fitting sample.
\Cref{cor:contextcost,prop:sharedcorrection} give
\begin{equation}
 \begin{split}
 \Delta_m^{\rm separate}
 &=\frac{z_\alpha\phi(z_\alpha)K}{m\kappa}+o(m^{-1}),\\
 \Delta_m^{\rm shared}
 &=z_\alpha\phi(z_\alpha)
   \left\{\frac1{m\kappa}
     +\frac{\Var_Q(\delta_C)}{v(1+\kappa)}\right\}
   +o(m^{-1}).
 \end{split}
 \label{eq:poolingsummary}
\end{equation}
Separate fitting incurs the estimation cost of $K$ changes.
Sharing reduces that cost to one change, while adding a penalty
for differences between the true changes. Both limiting intervals
have width $2z_\alpha\sqrt{v/(1+\kappa)}$, so their coverage difference
comes from the fitted centers. To first order, separate fitting
has the smaller gap when
\[
 \Var_Q(\delta_C)
 >\frac{v(1+\kappa)(K-1)}{m\kappa}.
\]
Thus, observing context does not by itself justify fitting a
different correction in every context. Similar changes favor
sharing, whereas sufficiently different changes favor separate fitting.
The expansions also hold after averaging over random context counts.
The comparison uses unknown changes and does not guarantee
coverage for a data-selected fit.
\Cref{app:contextcostproof,app:sharedproof} give the general
formulas and proofs. \Cref{sec:poolingexperiment} illustrates
the comparison.

\section{Experiments}
\label{sec:experiments}
We examine prediction benefits, the coverage cost of estimating
corrections, and performance on real measurements. Nominal coverage
is $0.90$, and weighted procedures use FWQ unless stated otherwise.
Target outcomes are excluded from adaptation except for the
target-labeled reference. Gaussian experiments use independent
samples. Air Quality results summarize day splits of one dependent
time series and do not establish a time-series coverage guarantee.

\paragraph{Methods compared.}
The suffix S denotes a score based on the source predictive
distribution, and T denotes one based on the target-adjusted
predictive distribution.
\textbf{Standard CP} uses source predictions and unweighted
split-conformal calibration on source data.
\textbf{Covariate-S} keeps the source predictions and weights
source calibration by an estimated target-to-source input density ratio.
SMM denotes stable measurement mechanisms.
\textbf{SMM-S} and \textbf{SMM-T} retain context and use the
estimated joint ratio in \eqref{eq:fittedproxyweight} under measurement
stability. For a given fitted correction, they use the same weights.
SMM-S adjusts calibration, while SMM-T also adjusts the predictive
distribution used for the score.

\subsection{Prediction benefit of measurement stability and context}
\label{sec:restrictionbenefit}
The source has $P(C=1)=1/2$ and
\[
 Y\mid C=c\sim\N(0,v_c),\qquad (v_0,v_1)=(0.5,2),\qquad
 Z=\lambda Y+\varepsilon,\quad\varepsilon\sim\N(0,1).
\]
The target changes to $Q(C=1)=3/4$ and
$Y\mid C=c\sim\N(2c-1,v_c)$ while preserving the measurement mechanism.
Each of $200$ paired replications uses $N=2000$ source training
pairs, $m=2000$ target inputs, and $n=6000$ calibration pairs.
The source mechanisms are estimated. All methods here use
unstandardized absolute residuals, so the different source
variances let covariate weighting change the threshold.
For the SMM methods, \textbf{shared} fits one common outcome
mean change across contexts, whereas \textbf{separate} fits a change
within each context. Both retain $C$ for prediction. Comparing
SMM-S and SMM-T with separate corrections isolates predictive adjustment.

\textbf{Pooled label-T} discards $C$ and assumes stability of
$Z\mid Y$. It fits a two-component target outcome mixture from $Z$
and weights source calibration by the estimated target-to-source
outcome density ratio. Its adjusted predictor uses $Z$ alone.
This family contains the true target distribution of $(Z,Y)$, so
the comparison isolates the value of retaining context.
Label-shift adjustment using $(C,Z)$ would instead require stability
of $(C,Z)\mid Y$. \textbf{Oracle ratio and mean} uses the true joint
ratio and target conditional mean, with finite source calibration
and the same FWQ rule. \Cref{app:gaussianprotocol} gives fitting details.

\begin{table}[htbp]\centering\small
\caption{Gaussian comparison at $\lambda=1$, with $200$ paired replications, $N=2000$, $m=2000$, $n=6000$, and nominal coverage $0.90$. All scores are unstandardized absolute residuals. Parentheses are Monte Carlo SEs. }
\label{tab:benefitmain}
\begin{tabular}{@{}lrr@{}}\toprule Method & Coverage & Mean length\\\midrule
Standard CP & 0.8094 (0.0006) & 2.327 (0.002)\\
Covariate-S & 0.8461 (0.0007) & 2.521 (0.003)\\
Pooled label-T & 0.9003 (0.0005) & 2.754 (0.004)\\
SMM-T, shared & 0.8391 (0.0006) & 2.522 (0.003)\\
SMM-S, separate & 0.9016 (0.0008) & 2.896 (0.007)\\
SMM-T, separate & 0.8996 (0.0006) & 2.521 (0.005)\\
Oracle ratio and mean & 0.9004 (0.0005) & 2.520 (0.004)\\
\bottomrule\end{tabular}\end{table}

In \Cref{tab:benefitmain}, SMM-T with separate corrections reaches
coverage $0.8996$, compared with $0.8461$ for Covariate-S at the
same mean length, $2.521$. Its intervals are about $8\%$ shorter
than Pooled label-T at comparable coverage. Predictive adjustment
also shortens intervals relative to SMM-S with the same weights.
A shared correction fails because the contexts change in opposite
directions.

Population residual quantiles give lengths $2.5150$ with context
and $2.7408$ without it, an $8.2\%$ reduction. This illustrates the
value of context, but the unequal target predictive variances
$1/3$ and $2/3$ fall outside the common-variance setting of
\Cref{prop:contextvalue}. We do not claim that these fixed-width
intervals attain its optimality bound.

As a non-Gaussian check, let the source outcome be $\N(0,1)$
in context zero and an equally weighted mixture
$\N(-1.5,0.09)$, $\N(1.5,0.09)$ in context one.
Use $Q(C=1)=0.2$, conditional tilt coefficients $(0.5,-0.5)$,
and $\lambda=0.5$. Covariate-S and SMM-T use standardized residuals, with SMM-T
fitting a separate outcome tilt per context. Covariate-S obtains
its input ratio from the fitted outcome and sensor models
(\Cref{app:gaussianprotocol}). Sample sizes are unchanged, with
an estimated shared Gaussian sensor model and $150$ paired fits.
Covariate-S gives coverage $0.8714$ (Monte Carlo SE $0.0006$) and
length $3.070$ (SE $0.004$), compared with $0.8982$ (SE $0.0006$)
and $2.902$ (SE $0.006$) for SMM-T. The source shapes,
as well as frequencies, now affect calibration.

\subsection{Sharing corrections and finite calibration}
\label{sec:poolingexperiment}\label{sec:finitecalibration}
To compare shared and separate corrections, use $K=16$ contexts,
source $Y\mid C=c\sim\N(0,1)$, and
$Z=0.5Y+\varepsilon$ with standard normal noise.
Target probabilities satisfy $\pi_c\propto0.8^c$, and outcome
means are $\delta_c=0.5+\eta u_c$, where
$\sum_c\pi_cu_c=0$ and $\sum_c\pi_cu_c^2=1$.
Thus $\eta=\operatorname{sd}_Q(\delta_C)$ measures differences between the changes.
We use $m=2000$, $n=6000$, and $200$ replications per $\eta$.
\Cref{fig:poolingcomparison} compares SMM-T with \emph{Shared}
and \emph{Separate} mean corrections. Both retain context and use
adjusted standardized residuals with FWQ. Source mechanisms are
known in panel (a) and estimated from $N=8000$ pairs in panel (b).
Both panels fit outcome changes from unlabeled target data.

\begin{figure}[tbp]
\centering
\begin{subfigure}[t]{.49\textwidth}\centering
\begin{tikzpicture}
\begin{axis}[width=0.97\linewidth,height=5.6cm,xmin=-0.01,xmax=0.43,ymin=0.872,ymax=0.902,
xlabel={Heterogeneity $\mathrm{sd}_Q(\delta_C)$},ylabel={Marginal coverage},yticklabel style={/pgf/number format/fixed,/pgf/number format/precision=3},
tick label style={font=\scriptsize},label style={font=\small},grid=major,grid style={gray!20},title={Known source},title style={font=\small},
legend style={font=\scriptsize,at={(0.03,0.03)},anchor=south west,draw=none,fill=white,fill opacity=0.85,text opacity=1}]
\addplot[gray,dotted,forget plot] coordinates {(-0.01,.9) (0.43,.9)};
\addplot[blue!65!black,thin] coordinates {(0.000,0.8946) (0.010,0.8946) (0.020,0.8946) (0.030,0.8946) (0.040,0.8946) (0.050,0.8946) (0.060,0.8946) (0.070,0.8946) (0.080,0.8946) (0.090,0.8946) (0.100,0.8946) (0.110,0.8946) (0.120,0.8946) (0.130,0.8946) (0.140,0.8946) (0.150,0.8946) (0.160,0.8946) (0.170,0.8946) (0.180,0.8946) (0.190,0.8946) (0.200,0.8946) (0.210,0.8946) (0.220,0.8946) (0.230,0.8946) (0.240,0.8946) (0.250,0.8946) (0.260,0.8946) (0.270,0.8946) (0.280,0.8946) (0.290,0.8946) (0.300,0.8946) (0.310,0.8946) (0.320,0.8946) (0.330,0.8946) (0.340,0.8946) (0.350,0.8946) (0.360,0.8946) (0.370,0.8946) (0.380,0.8946) (0.390,0.8946) (0.400,0.8946) (0.410,0.8946) (0.420,0.8946) (0.430,0.8946) (0.440,0.8946) (0.450,0.8946)};
\addlegendentry{Separate population limit}
\addplot[red!70!black,thin,dashed] coordinates {(0.000,0.8997) (0.010,0.8996) (0.020,0.8996) (0.030,0.8995) (0.040,0.8994) (0.050,0.8993) (0.060,0.8992) (0.070,0.8990) (0.080,0.8988) (0.090,0.8986) (0.100,0.8983) (0.110,0.8980) (0.120,0.8977) (0.130,0.8974) (0.140,0.8970) (0.150,0.8966) (0.160,0.8962) (0.170,0.8957) (0.180,0.8953) (0.190,0.8948) (0.200,0.8943) (0.210,0.8937) (0.220,0.8931) (0.230,0.8925) (0.240,0.8919) (0.250,0.8912) (0.260,0.8905) (0.270,0.8898) (0.280,0.8891) (0.290,0.8883) (0.300,0.8876) (0.310,0.8867) (0.320,0.8859) (0.330,0.8850) (0.340,0.8842) (0.350,0.8832) (0.360,0.8823) (0.370,0.8813) (0.380,0.8804) (0.390,0.8793) (0.400,0.8783) (0.410,0.8773) (0.420,0.8762) (0.430,0.8751) (0.440,0.8739) (0.450,0.8728)};
\addlegendentry{Shared population limit}
\addplot[blue!65!black,only marks,mark=*,mark size=1.6pt] coordinates {(0.000,0.8938) (0.050,0.8944) (0.100,0.8947) (0.200,0.8942) (0.400,0.8942)};
\addlegendentry{Separate}
\addplot[red!70!black,only marks,mark=square*,mark size=1.6pt] coordinates {(0.000,0.8993) (0.050,0.8996) (0.100,0.8986) (0.200,0.8937) (0.400,0.8785)};
\addlegendentry{Shared}
\end{axis}
\end{tikzpicture}
\caption{}
\end{subfigure}\hfill
\begin{subfigure}[t]{.49\textwidth}\centering
\begin{tikzpicture}
\begin{axis}[width=0.97\linewidth,height=5.6cm,xmin=-0.01,xmax=0.43,ymin=0.872,ymax=0.902,
xlabel={Heterogeneity $\mathrm{sd}_Q(\delta_C)$},ylabel={Marginal coverage},yticklabel style={/pgf/number format/fixed,/pgf/number format/precision=3},
tick label style={font=\scriptsize},label style={font=\small},grid=major,grid style={gray!20},title={Estimated source ($N=8000$)},title style={font=\small},
legend style={font=\scriptsize,at={(0.03,0.03)},anchor=south west,draw=none,fill=white,fill opacity=0.85,text opacity=1}]
\addplot[gray,dotted,forget plot] coordinates {(-0.01,.9) (0.43,.9)};
\addplot[blue!65!black,only marks,mark=*,mark size=1.6pt] coordinates {(0.000,0.8914) (0.050,0.8920) (0.100,0.8920) (0.200,0.8927) (0.400,0.8932)};
\addlegendentry{Separate}
\addplot[red!70!black,only marks,mark=square*,mark size=1.6pt] coordinates {(0.000,0.8987) (0.050,0.8986) (0.100,0.8980) (0.200,0.8945) (0.400,0.8790)};
\addlegendentry{Shared}
\end{axis}
\end{tikzpicture}
\caption{}
\end{subfigure}
\caption{SMM-T with shared and separate location corrections, $K=16$ contexts, $v=1$, $\lambda=0.5$ ($\kappa=0.25$), $m=2000$ target inputs and unequal context frequencies $Q_C(c)\propto0.8^{c}$. Points are means over $200$ replications with $n=6000$. Lines in (a) evaluate the large-calibration formulas \eqref{eq:contextlocationcost} and \eqref{eq:sharedcoverage} at the expected counts. \eqref{eq:poolingtradeoff} gives the first-order crossing $\mathrm{sd}_Q(\delta_C)=\{v(K-1)(1+\kappa)/(m\kappa)\}^{1/2}=0.194$.}
\label{fig:poolingcomparison}
\end{figure}
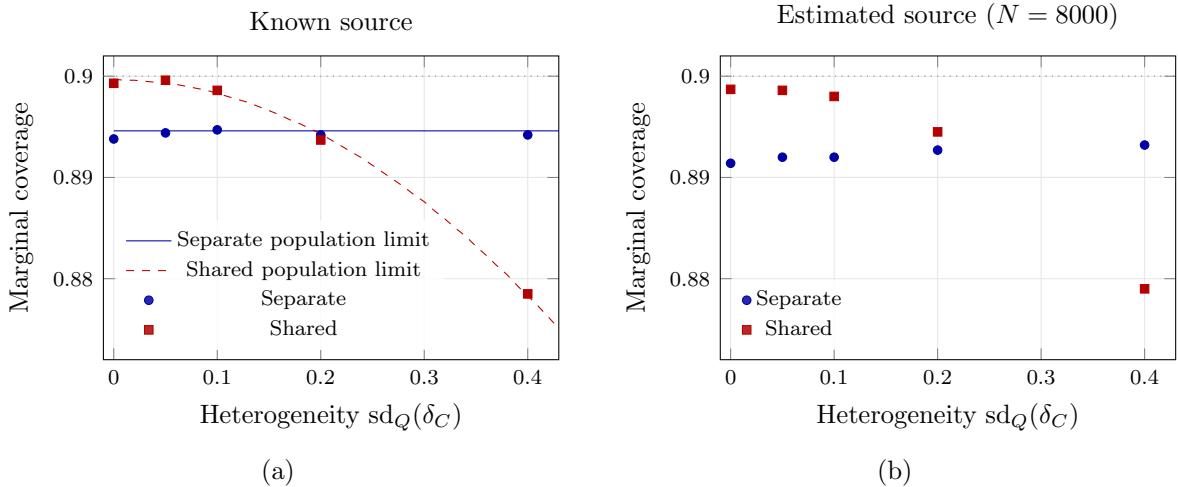

Sharing gives higher coverage when changes are similar, while
separate fitting avoids the bias from sharing when they differ.
At expected counts, the curves cross at $\eta=0.1931$, close to
the first-order value $0.194$. With estimated source mechanisms,
separate-fit coverage rises slightly from $0.8914$ to $0.8932$.
The known-source analysis does not explain this trend, and without
Monte Carlo SEs we do not interpret it as systematic.

Increasing calibration size does not remove target-fitting error.
We apply FWQ to the Gaussian mean-and-variance correction in
\Cref{subsec:scalarlearning}, with adjusted standardized residuals
and known source mechanisms. With $\lambda=0.25$, $m=2000$, and
true outcome variance one, increasing $n$ from $1500$ to $96000$
moves coverage from $0.8359$ to $0.8384$, toward its
large-calibration limit $0.8386$, rather than toward $0.90$.
The corresponding values at variance $1.5$ are $0.8636$,
$0.8708$, and $0.8713$.
\Cref{app:fittingprotocol} gives the finite-calibration protocol.

\subsection{Air Quality under controlled and calendar-based shifts}
\label{sec:air}

We use $7344$ complete hourly observations on $347$ days from
the UCI Air Quality dataset \citep{DeVitoS2008sab}, whose documentation
reports concept and sensor drift \citep{DeVitoS2008uci}.
The outcome is log carbon monoxide (CO), the measurements are five
log sensor readings, and four contexts split temperature at
$20^\circ$C and relative humidity at $50\%$. The source model has
context-specific outcome location and scale, a shared residual
mixture, and a quadratic Gaussian sensor model. Each design uses
$50$ day splits, with length measured in log CO.
\Cref{app:airprotocol} gives fitting and evaluation details.

Calendar-based source--target splits use different day types or
observation periods for the two populations. The source-to-target
directions are weekdays to weekends, March--May to
September--November 2004 (spring to autumn), and March 10--April 30,
2004 to March 1--April 4, 2005 (year over year). The year-over-year
comparison thus uses unequal date ranges.
The source group supplies labeled training and calibration data.
The target group supplies unlabeled fitting data and separate days
for evaluation. Unlike the controlled selection below, these splits
do not enforce measurement stability.

We compare Standard CP, Covariate-S, and SMM-T using
standardized residuals.
For these calendar-based splits, the \textbf{Target-labeled reference}
keeps the source predictor but uses unweighted calibration on
outcomes from target fitting days. It therefore uses labels
unavailable to the adaptation methods.

SMM-T fits one bounded exponential tilt per context using
log CO clipped to $[\log0.3,\log8]$. The main controlled selection
uses $\ell(y)=1.6\tanh\{(y-\log2)/1.6\}$, with target-to-source
density ratio proportional to
\[
 \exp\{0.25s_T+1.2s_H\ell(y)\},
\]
where $s_T,s_H\in\{-1,1\}$ indicate warm and humid conditions.
Selection on outcome and context preserves the measurement
mechanism, although the different selection statistic places the
true correction outside the fitted tilt family. No-shift and
context-only designs provide reference cases.

\begin{table}[ht!]\centering\small\setlength{\tabcolsep}{7pt}
\caption{Air Quality prediction with one source-model family and
$50$ day splits per design. Entries are coverage / mean length in
log CO. Nominal coverage is $0.90$. Scores are standardized
residuals, and weighted procedures use FWQ. Only the target-labeled
reference uses outcomes from target fitting days.}
\label{tab:airmain}
\textbf{(a) Controlled shifts}\par\smallskip
\begin{tabular}{@{}lrrr@{}}\toprule
Shift & Standard CP & Covariate-S & SMM-T\\\midrule
No shift & 0.894 / 0.77 & 0.895 / 0.77 & 0.894 / 0.77\\
Context only & 0.875 / 0.75 & 0.893 / 0.81 & 0.896 / 0.83\\
Context and outcome & 0.861 / 0.68 & 0.885 / 0.75 & 0.901 / 0.85\\
\bottomrule\end{tabular}
\par\medskip
\textbf{(b) Calendar-based source--target splits}\par\smallskip
\begin{tabular}{@{}lrrr@{}}\toprule
Method & Weekday--weekend & Spring--autumn & Year over year\\\midrule
Standard CP & 0.867 / 0.77 & 0.751 / 0.54 & 0.437 / 0.37\\
Covariate-S & 0.887 / 0.85 & 0.733 / 0.51 & 0.451 / 0.39\\
SMM-T & 0.894 / 0.93 & 0.708 / 0.50 & 0.491 / 0.48\\
Target-labeled reference & 0.892 / 0.88 & 0.899 / 0.86 & 0.886 / 1.29\\
\bottomrule\end{tabular}
\end{table}

Under controlled context-and-outcome shift (\Cref{tab:airmain}),
SMM-T raises coverage from $0.885$ for Covariate-S to $0.901$
at a length cost, from $0.75$ to $0.85$. The paired coverage
difference is $0.0163$ (split Monte Carlo SE $0.0036$).

For weekday-to-weekend transfer, SMM-T reaches coverage $0.894$
versus $0.887$ for Covariate-S, but its intervals are longer than
the target-labeled reference's ($0.93$ versus $0.88$).
For spring-to-autumn transfer, adaptation worsens coverage to
$0.708$, below Standard CP ($0.751$) and Covariate-S ($0.733$),
with slightly shorter intervals ($0.50$ versus $0.54$ and $0.51$).
Year-over-year coverage is only $0.491$. The target-labeled reference
reaches $0.899$ and $0.886$ in the latter two transfers, using
labels unavailable to the adaptation methods.

Context discretization is itself a possible source of instability.
If $W$ denotes continuous weather and $C=h(W)$ its four-bin
summary, stability given $(Y,W)$ yields
\[
 q(z\mid y,c)=\int k(z\mid y,w)\,q(\dd w\mid y,c).
\]
Changes in weather composition within a bin can therefore change
$Z\mid Y,C$ even with unchanged physical sensor responses.
The failures do not distinguish this effect from sensor changes,
model misspecification, temporal dependence, or limited overlap.

\section{Conclusion}
\label{sec:discussion}

The stable measurement model gives an application-based
interpretation of sparse joint shift. Our analysis separates the
predictive benefit of context from the coverage cost of estimating
outcome changes. With known source mechanisms, Gaussian formulas
quantify coverage gaps from a finite unlabeled target sample even
as calibration size tends to infinity. Weak measurements amplify
the coverage cost of variance estimation. Sharing a correction
reduces fitting error when changes are similar, while sufficiently
different changes favor separate corrections.

In the synthetic comparisons, SMM-T improves coverage over covariate
weighting and shortens intervals relative to a correctly specified
pooled label correction that omits context. Controlled Air Quality
shifts yield coverage gains with longer intervals. Calendar-based
shifts reveal substantial undercoverage, with adaptation worsening
coverage on the spring-to-autumn transfer. Stable physical sensors
need not satisfy the restriction conditional on a noisy outcome or
coarsened context. Matching target inputs cannot verify that restriction.

Exact weighted ranks guarantee finite-sample marginal coverage with
the true ratio under the stated sampling assumptions. FWQ introduces
quantile-approximation error, and fitted ratios introduce
distributional error. Extending the coverage analysis to estimated
measurement mechanisms, broader outcome families, and dependent
observations remains future work.

\FloatBarrier
\begingroup
\bibliographystyle{plainnat}
\bibliography{sjc}
\endgroup

\clearpage
\appendix
\section{Proofs and supplementary calculations}
\label{app:proofs}

The proofs follow the order of the results discussed in
\Cref{sec:gaussian}. We also give the exact Gaussian evaluation
formulas and the full statements of the context-fitting
results summarized in the main text. The identification lemma is proved immediately after its statement in the main text.

\subsection{Proof of Lemma~\ref*{prop:fwq} (finite-quantile discretization)}
\label{app:fwqproof}
\begin{proof}
Condition on $\calA$ and on the unordered collection of all
$n+1$ observations. By weighted exchangeability
\citep{TibshiraniR2019neurips}, the probability that observation
$k$ is the target draw is
$p_k=w_k/\sum_{j=1}^{n+1}w_j$.
We bound the total probability of indices that the FWQ rule
fails to cover when treated as the target observation.

If deleting index $k$ leaves zero calibration weight, the
procedure returns $\calY$ and covers its outcome. Otherwise,
failure means that $S_k$ exceeds the weighted calibration
quantile. Consequently,
\[
 \sum_{j\ne k:S_j\ge S_k}w_j
 \le\alpha\sum_{j\ne k}w_j,
 \qquad
 \sum_{j\ne k:S_j\ge S_k}p_j\le\alpha.
\]
Let $\mathcal K$ contain the indices satisfying the second
inequality. Every failure index belongs to $\mathcal K$.
If this set is nonempty, choose $k^*\in\mathcal K$ with the
smallest score. Every other member of $\mathcal K$ has score
at least $S_{k^*}$, including ties, so
\[
 \sum_{k\in\mathcal K}p_k
 \le\alpha+p_{k^*}
 \le\alpha+\max_kp_k.
\]
If $\mathcal K$ is empty, miscoverage is zero. Averaging over
the unordered observations proves \eqref{eq:fwqcoverage}.
\end{proof}

\subsection{Proof of Corollary~\ref*{cor:fwqfitted} (fitted finite quantiles)}
\label{app:fittedfwqproof}
\begin{proof}
Condition on $\calA$, so the score and fitted weight function
are fixed. Under the distribution $R_{\widehat w}$ in
\eqref{eq:effective}, the normalized weight
$\overline w=\widehat w/\E_P\widehat w$ is the true ratio
relative to $P$. Multiplying all weights by the same constant
changes neither the FWQ set nor the largest normalized weight.
\Cref{prop:fwq} therefore gives
\[
 \Prob_{P^n\otimes R_{\widehat w}}
 \{Y_*\in\calC_{\rm fwq}(X_*)\mid\calA\}
 \ge1-\alpha-
 \E_{P^n\otimes R_{\widehat w}}
 [\widehat\pi_{\max}\mid\calA].
\]
The calibration distribution is the same under both product
distributions, and hence
\[
 \TV(P^n\otimes Q,P^n\otimes R_{\widehat w})
 =\TV(Q,R_{\widehat w})
 =\tfrac12\E_P|w-\overline w|.
\]
Changing the test distribution from $R_{\widehat w}$ to $Q$
can decrease the coverage probability by at most this amount.
This is the estimated-weight transfer argument of
\citet{LeiL2021jrsssb}, applied to the joint prediction pair,
and it proves \eqref{eq:fittedfwq}. Any score or model-family
choice based only on $\calA$ remains fixed throughout the argument.
\end{proof}

\subsection{Proof of Proposition~\ref*{prop:contextvalue} (context and minimum interval length)}
\label{app:contextproof}
\begin{proof}
At a fixed input, a normal distribution assigns the greatest
probability to a fixed-length interval centered at its mean.
For length $\ell\ge0$, that probability is
$f(\ell)=2\Phi\{\ell/(2\sigma)\}-1$, the same function at
every input. It is increasing and strictly concave on
$[0,\infty)$. Writing $L=|I(C,Z)|$, marginal coverage and
Jensen's inequality give
\[
 1-\alpha\le\Prob_Q\{Y\in I(C,Z)\}
 \le\E_Q f(L)\le f(\E_Q L).
\]
Since $f(2\sigma z_\alpha)=1-\alpha$, this proves the bound.
Intervals centered at $\mu_{\rm pred}(C,Z)$ with constant
length $2\sigma z_\alpha$ attain it.

For a $Z$-only interval of nonconstant length, Jensen's inequality
is strict. If its length is constant, positive required coverage
makes that length positive. The condition
$\Var_Q\{\mu_{\rm pred}(C,Z)\mid Z\}>0$ on a set of positive
probability prevents its $Z$-measurable midpoint from agreeing with
the conditional mean almost surely. Optimal centering is then
strict after averaging. In either case
$1-\alpha<f(\E_Q L)$, proving strictly greater expected length.
\end{proof}

\subsection{Proof of Proposition~\ref*{prop:scaleexact} (exact coverage and oracle gap)}
\label{app:exactproof}

The expectation in \eqref{eq:scaleaverage} reduces to a
one-dimensional integral. Let
$a_m^2=\sigma_Z^2/(m\lambda^2)$ and $\nu=m-1$.
For a generic fitted variance $v_0>0$, define
\begin{equation}
 \begin{split}
 T(x)&=\left[\frac{\sigma_Z^2x/m-\tau^2}{\lambda^2}
       \right]_{[v_{\min},v_{\max}]},\\
 \mathcal H_v(v_0)
 &=2\Phi\left(z_\alpha\sqrt{
 \frac{v_0(\lambda^2v_0+\tau^2)}
      {\tau^2(v+a_m^2)+\lambda^2v_0^2}}
 \right)-1.
 \end{split}
 \label{eq:Hscale}
\end{equation}
Then the averaged coverage and its gap from the oracle are
\begin{equation}
 \begin{split}
 C_m(\lambda,v)
 &=\E_{U\sim\chi^2_\nu}[\mathcal H_v\{T(U)\}],\\
 \Delta_m
 &=\E_{U\sim\chi^2_\nu}[(1-\alpha)-\mathcal H_v\{T(U)\}].
 \end{split}
 \label{eq:scaleexact}
\end{equation}
The function $T$ includes the fixed variance constraints.

\begin{proof}
\emph{Calibration under the fitted model.}
Condition on $(\widehat\mu,\widehat v)$. The shared measurement
density gives $\widehat w(y)p_Y(y)p(z\mid y)
=\widehat q_Y(y)q(z\mid y)$. Thus $\E_P\widehat w=1$, and
weighting the source gives the fitted target distribution.
Its standardized predictive residual is $\N(0,1)$, so the
absolute-score distribution function is
$G(t)=2\Phi(t)-1$ for $t\ge0$. The weighted strong law and
continuity of $G$ give convergence of the FWQ threshold to
$z_\alpha$.
For each fixed candidate $(z,y)$, its normalized weight tends
to zero, and the exact weighted upper-tail probability in
\eqref{eq:wpvalue} converges to $1-G\{\widehat s(z,y)\}$.
The boundary $\widehat s(Z,Y)=z_\alpha$ has zero probability
under $Q$. Bounded convergence over the calibration sample
and new target pair therefore gives the same limiting
coverage for both rules, namely the true target coverage
of $I_\infty$.

\emph{Coverage under the true target model.}
Substituting $Z=\lambda Y+\varepsilon$ into the predictive
mean gives the residual
\[
 R=Y-\widehat\mu_{\rm pred}(Z)
  =(1-\lambda\widehat\gamma)(Y-\widehat\mu)
   -\widehat\gamma\varepsilon.
\]
Under $Q$, $Y\sim\N(\mu,v)$ and
$\varepsilon\sim\N(0,\tau^2)$ are independent. Hence
$R\mid\widehat\mu,\widehat v\sim\N(b,\sigma_{\rm res}^2)$, with $b$ and
$\sigma_{\rm res}^2$ in \eqref{eq:conditionalcoverage}. Evaluating
$\Prob_Q\{|R|\le z_\alpha\sqrt{\widehat v_{\rm pred}}
\mid\widehat\mu,\widehat v\}$ gives that equation.
For the oracle, setting the fitted parameters to their true
values gives $b=0$ and $\sigma_{\rm res}^2=v_{\rm pred}$, hence coverage
$1-\alpha$. Coverage is bounded by one, so bounded convergence
also justifies averaging over the fitting sample, as in
\eqref{eq:scaleaverage}.

\emph{Averaging over the fitted parameters.}
The normal sampling identities \citep{CochranW1934mpcps} give
\[
 \widehat\mu-\mu\sim\N(0,a_m^2),\qquad
 U=mS_Z^2/\sigma_Z^2\sim\chi^2_\nu,
\]
with independent terms, and $\widehat v=T(U)$. Conditional
on $\widehat v$, averaging over $\widehat\mu$ gives a normal
residual with mean zero and variance
$\sigma_{\rm res}^2+(1-\lambda\widehat\gamma)^2a_m^2$. Consequently,
\[
 \frac{\widehat v_{\rm pred}}{\sigma_{\rm res}^2+(1-\lambda\widehat\gamma)^2a_m^2}
 =\frac{\widehat v(\lambda^2\widehat v+\tau^2)}
 {\tau^2(v+a_m^2)+\lambda^2\widehat v^2}.
\]
This gives $\mathcal H_v(\widehat v)$ in \eqref{eq:Hscale}.
Averaging over $U$ gives $C_m$, and subtracting from $1-\alpha$
gives $\Delta_m$ in \eqref{eq:scaleexact}.
The centered mean error and variance pivot have distributions
independent of $\mu$, explaining the same property of both quantities.

Fixing the fitted variance at $v_0$ gives $\mathcal H_v(v_0)$.
For $v_0=v$, we have $\sigma_{\rm res}^2=v_{\rm pred}=v\tau^2/(\lambda^2v+\tau^2)$,
and averaging over the fitted mean gives
\[
 \frac{R}{\sqrt{v_{\rm pred}}}
 \sim\N\!\left(0,1+\frac{\tau^2}{m\lambda^2v}\right).
\]
Its probability of lying between $-z_\alpha$ and $z_\alpha$
is $C_m^{\rm loc}$ in \eqref{eq:locationlimit}. The added
variance is positive for finite $m$, so this probability is
strictly below $1-\alpha$ and $\Delta_m^{\rm loc}>0$.
\end{proof}

Clipping can yield overcoverage even when the true variance is
inside the bounds. For $\alpha=0.10$, $m=100$, $\lambda=0.5$,
$v=\tau^2=1$, and bounds $[0.99,10]$, numerical evaluation of
\eqref{eq:scaleexact} gives $C_m=0.9103$ and $\Delta_m=-0.0103$.

\subsection{Proof of Proposition~\ref*{prop:scalecost} (leading coverage gaps)}
\label{app:scaleproof}
\paragraph{Fisher information for outcome parameters.}
The target sensor distribution also shows why its mean and
variance are not equally easy to estimate. Its Fisher
information per observation is
\begin{equation}
 I_{\mu\mu}=\frac{\lambda^2}{\sigma_Z^2},\qquad
 I_{vv}=\frac{\lambda^4}{2\sigma_Z^4},\qquad
 I_{\mu v}=0.
 \label{eq:information}
\end{equation}
With $v$ and $\tau^2$ fixed, information about the outcome
mean decreases as $\lambda^2$ when the loading approaches
zero, whereas information about the outcome variance decreases as
$\lambda^4$. At zero loading, the sensor contains no
information about either parameter. This information loss explains the sensitivity to weak
measurements in \Cref{subsec:scalarlearning}.

\paragraph{Coverage curvature and variance-estimation error.}
To understand the additional cost of estimating variability,
first suppose that the target mean is known. Let
$\kappa=\lambda^2v/\tau^2$ be the measurement signal-to-noise
ratio, and let $\rho>0$ denote the fitted-to-true outcome variance ratio.
The coverage as a function of $\rho$ is
\[
 h_\kappa(\rho)
 =2\Phi\left(z_\alpha\sqrt{
   \frac{\rho(1+\kappa\rho)}{1+\kappa\rho^2}}
   \right)-1.
\]
At the true variance,
\begin{equation}
 \begin{split}
 h_\kappa'(1)
 &=\frac{z_\alpha\phi(z_\alpha)}{1+\kappa},\\
 h_\kappa''(1)
 &=-\frac{z_\alpha\phi(z_\alpha)}{(1+\kappa)^2}
   \left\{\frac{1+z_\alpha^2}{2}+4\kappa\right\}<0.
 \end{split}
 \label{eq:curvature}
\end{equation}
The negative curvature means that small, mean-zero errors
in the variance estimate reduce average coverage to second
order. Near the truth, the coverage lost by underestimating
the variance exceeds the gain from overestimating it by
the same amount. This local Jensen effect does not require
global concavity.

The size of the variance-estimation error can be seen directly.
Before projection, write
$\widetilde v=(S_Z^2-\tau^2)/\lambda^2$. Its exact moments are
\begin{equation}
 \E(\widetilde v/v-1)=-\frac{1+\kappa}{m\kappa},\qquad
 \Var(\widetilde v/v)
 =\frac{2(m-1)(1+\kappa)^2}{m^2\kappa^2}.
 \label{eq:scalemoments}
\end{equation}

\begin{proof}
Write $\rho=\widehat v/v$, $e=a_m^2/v=(1+\kappa)/(m\kappa)$,
and
\[
 G_\kappa(\rho,e)
 =2\Phi\left(z_\alpha\sqrt{
 \frac{\rho(1+\kappa\rho)}{1+e+\kappa\rho^2}}\right)-1.
\]
By \Cref{prop:scaleexact}, $C_m=\E G_\kappa(\rho,e)$.
Since $G_\kappa(1,0)=1-\alpha$, the gap to be expanded is
\[
 \Delta_m=G_\kappa(1,0)-\E G_\kappa(\rho,e).
\]
At $(\rho,e)=(1,0)$, the required derivatives are
\begin{align*}
 G_\kappa(1,0)&=1-\alpha, &
 \partial_\rho G_\kappa(1,0)&=\frac{z_\alpha\phi(z_\alpha)}{1+\kappa},\\
 \partial_e G_\kappa(1,0)&=-\frac{z_\alpha\phi(z_\alpha)}{1+\kappa}, &
 \partial_{\rho\rho}G_\kappa(1,0)&=h_\kappa^{\prime\prime}(1).
\end{align*}
Let $\widetilde\rho=\widetilde v/v$ and choose $\epsilon>0$
so that $[v(1-\epsilon),v(1+\epsilon)]$ lies strictly inside
the clipping bounds. On $A_m=\{|\widetilde\rho-1|\le\epsilon\}$,
clipping is inactive and $\rho=\widetilde\rho>0$.
Chi-square concentration gives $\Prob(A_m^c)\le Ce^{-cm}$
for constants $C,c>0$. Since $0\le G_\kappa(\rho,e)\le1$,
the coverage contribution on $A_m^c$ is exponentially small.
Cauchy--Schwarz gives the bound
\[
 \E\bigl[|\widetilde\rho-1|^j\mathbf1_{A_m^c}\bigr]
 \le\bigl\{\E|\widetilde\rho-1|^{2j}\bigr\}^{1/2}
      \Prob(A_m^c)^{1/2}=o(m^{-1}),\qquad j=1,2.
\]
The clipped $\rho$ is bounded, so its first two moments agree
with those of $\widetilde\rho$ up to $o(m^{-1})$.
Taylor expansion is used only on $A_m$, where derivatives
are bounded. Its expected remainder is
$O(\E|\widetilde\rho-1|^3+e\E|\widetilde\rho-1|+e^2)
=O(m^{-3/2})$. Thus \eqref{eq:scalemoments} gives
\[
 \begin{split}
 \E G_\kappa(\rho,e)
 &=G_\kappa(1,0)
 +\partial_\rho G_\kappa(1,0)\E(\rho-1)
 +\partial_e G_\kappa(1,0)e\\
 &\quad+\tfrac12\partial_{\rho\rho}G_\kappa(1,0)
 \Var(\rho)+o(m^{-1}).
 \end{split}
\]
The three corrections after $1-\alpha$ are, respectively,
\[
 -\frac{z_\alpha\phi(z_\alpha)}{m\kappa},\qquad
 -\frac{z_\alpha\phi(z_\alpha)}{m\kappa},\qquad
 -\frac{z_\alpha\phi(z_\alpha)}m
 \left\{\frac4\kappa+\frac{1+z_\alpha^2}{2\kappa^2}\right\}.
\]
They arise from variance-estimator bias, mean estimation, and
curvature with respect to the variance estimate. Subtracting
the resulting coverage from $1-\alpha$ gives the expansion of
$\Delta_m$ in \eqref{eq:scalecost}. Using the unbiased sample variance
removes the first contribution without changing the leading
variance term.

Expanding \eqref{eq:locationlimit} gives the expansion of $\Delta_m^{\rm loc}$
in \eqref{eq:scalecost}. Both leading gaps are positive, and
$\Delta_m-\Delta_m^{\rm loc}$ has positive leading coefficient
$z_\alpha\phi(z_\alpha)\{5/\kappa+(1+z_\alpha^2)/(2\kappa^2)\}$.
This proves $0<\Delta_m^{\rm loc}<\Delta_m$ for sufficiently large $m$.
\end{proof}

\paragraph{Unbiased variance estimation and the numerical example.}
If the variance estimate instead uses divisor $m-1$ in
$S_Z^2$ before projection, the coefficient $6/\kappa$
becomes $5/\kappa$. For the numerical example in \Cref{subsec:scalarlearning}, the unprojected maximum-likelihood estimate has
$\Var(\widetilde v/v)=0.2889$,
and projection puts $\widehat v$ at the lower bound $0.05$
with probability $0.0380$. Using the unbiased sample variance
before projection gives coverage $0.8407$, so removing the
bias leaves most of the loss. These coverage values are
computed from the exact expectation, with the appropriate
sample-variance divisor and projection. The expansion holds
at fixed $\kappa$ as $m$ grows. The exact expression should
be used when weak measurements produce large errors or
frequent estimates at the bounds.

\subsection{Statement and proof of Corollary~\ref*{cor:contextcost} (context-specific fitting cost)}
\label{app:contextcostproof}

The following corollary applies the scalar coverage formulas
within each context and averages using the target context
probabilities. It is a direct consequence of
\Cref{prop:scaleexact,prop:scalecost}.

\begin{corollary}[Context-specific fitting cost]
\label{cor:contextcost}
Consider $K$ contexts with target probabilities $\pi_c>0$.
The source context probabilities and conditional Gaussian
outcome distributions are known. In target context $c$, suppose
\[
 Y\sim\N(\mu_c,v_c),\qquad
 Z=\lambda_cY+\varepsilon_c,\qquad
 \varepsilon_c\sim\N(0,\tau_c^2),\qquad
 \kappa_c=\lambda_c^2v_c/\tau_c^2>0,
\]
with known shared measurement parameters and noise independent
of $Y$. Let $M_c$ count the unlabeled target observations
in context $c$. Fit a separate mean in each context and use
the resulting joint ratio and adjusted standardized Gaussian
score. Conditional on $M_c=m_c\ge2$ for every $c$, if the
outcome variances are known, the large-calibration coverage
averaged over the target fitting observations is
\begin{equation}
 \sum_{c=1}^K\pi_c\left[
 2\Phi\left\{\frac{z_\alpha}
 {\sqrt{1+1/(m_c\kappa_c)}}\right\}-1
 \right].
 \label{eq:contextlocationcost}
\end{equation}
If each variance is also fitted by the projected
maximum-likelihood estimator, the exact expression is
$\sum_c\pi_c C_{m_c}(\lambda_c,v_c)$, with each context's
noise variance and fixed bounds used in \eqref{eq:scaleexact}.

Let $\Delta_m^{\rm separate}$ and
$\Delta_m^{\rm separate,loc+scale}$ denote $1-\alpha$ minus
the averaged limiting coverage for separate mean-only and
mean-and-variance fits, respectively. For fixed $K$, fixed
positive $\pi_c$ and $\kappa_c$, and true variances strictly
inside their bounds, averaging also over the random context
counts gives
\begin{align}
 \Delta_m^{\rm separate}
 &=\frac{z_\alpha\phi(z_\alpha)}{m}
       \sum_{c=1}^K\frac1{\kappa_c}+o(m^{-1}),\notag\\
 \Delta_m^{\rm separate,loc+scale}
 &=\frac{z_\alpha\phi(z_\alpha)}{m}
       \sum_{c=1}^K\left\{\frac6{\kappa_c}
                    +\frac{1+z_\alpha^2}{2\kappa_c^2}\right\}
       +o(m^{-1}).
 \label{eq:contextasymptoticcost}
\end{align}
Here $n\to\infty$ before $m\to\infty$. Any strictly positive
fitted context probabilities fixed by $\calA$ may be used
in the weights. Use a prespecified fallback if any $M_c<2$.
These events have exponentially small probability in this
asymptotic regime.
\end{corollary}

To interpret the result, context $c$ receives about
$m\pi_c$ target observations, giving a leading location
loss proportional to $1/(m\pi_c\kappa_c)$ within that context.
Marginal coverage weights this loss by $\pi_c$, leaving the
term $1/(m\kappa_c)$ in \eqref{eq:contextasymptoticcost}.
Rare contexts can therefore have poor coverage individually,
even though their probabilities cancel from the leading
marginal expression. The proof below
also explains why fitted context probabilities do not affect
this particular population threshold.

\begin{proof}
Conditional on the context counts, each context has the normal
mean and variance pivots used in \Cref{prop:scaleexact}. Under
the fitted target distribution, the adjusted score is absolute
standard normal in every context. Its mixture over contexts
therefore has the same distribution, regardless of the positive
fitted context probabilities. The population weighted quantile
is consequently $z_\alpha$.

Apply the scalar location-only formula within each context
and average using the true probabilities $\pi_c$ to obtain
\eqref{eq:contextlocationcost}. When variances are also fitted,
the same argument gives $\sum_c\pi_cC_{m_c}(\lambda_c,v_c)$.

For multinomial counts with fixed $\pi_c>0$,
\[
 \E\left[\frac{\ind\{M_c\ge2\}}{M_c}\right]
 =\frac1{m\pi_c}+O(m^{-2}),
\]
where the fraction is defined as zero when $M_c<2$, and
$\Prob(M_c<2)$ decays exponentially in $m$.
Apply the scalar expansions for counts near $m\pi_c$.
Multiplication by $\pi_c$ cancels the corresponding denominator,
giving \eqref{eq:contextasymptoticcost}. Counts away from this
range have exponentially small probability, and coverage is
bounded between zero and one. The prespecified fallback therefore
does not affect the displayed first-order terms.
\end{proof}

\subsection{Statement and proof of Proposition~\ref*{prop:sharedcorrection} (shared versus separate corrections)}
\label{app:sharedproof}

We state the Gaussian comparison for potentially different
source means, variances, and measurement strengths across
contexts. The result uses classical normal sampling identities
\citep{CochranW1934mpcps}.

\bigskip
\begin{proposition}[Shared versus separate location corrections]
\label{prop:sharedcorrection}
Let the source outcome mean in context $c$ be $\mu_c^P$ and
the target mean be $\mu_c^P+\delta_c$, with the same known
variance $v_c$ in both populations. Retain the setting of
\Cref{cor:contextcost}, allowing a known sensor intercept in
$Z=a_c+\lambda_cY+\varepsilon_c$.
Condition on target counts $m_c\ge2$. The separate change
estimates and their precisions are
\[
 \widehat\delta_c
 =\frac{\overline Z_c^{\rm t}-a_c}{\lambda_c}-\mu_c^P,
 \qquad
 J_c=\frac{m_c\kappa_c}{v_c(1+\kappa_c)}.
\]
The maximum-likelihood estimate of a shared change is
\[
 \widehat\delta=\frac{\sum_cJ_c\widehat\delta_c}{\sum_cJ_c},
 \qquad
 \overline\delta=\frac{\sum_cJ_c\delta_c}{\sum_cJ_c},
 \qquad
 v_\delta=\left(\sum_cJ_c\right)^{-1},
\]
so that $\widehat\delta\sim\N(\overline\delta,v_\delta)$.
Use the ratio and adjusted standardized score from this shared fit.
As $n\to\infty$, its target coverage, averaged over the fitting sample, is
\begin{equation}
 \begin{split}
 C_{\rm shared}
 &=\sum_c\pi_c\left[
 \Phi\left(\frac{z_\alpha-b_c}{\sqrt{1+\omega_c}}\right)
 -\Phi\left(\frac{-z_\alpha-b_c}{\sqrt{1+\omega_c}}\right)
 \right],\\
 b_c&=\frac{\delta_c-\overline\delta}{\sqrt{v_c(1+\kappa_c)}},
 \qquad
 \omega_c=\frac{v_\delta}{v_c(1+\kappa_c)}.
 \end{split}
 \label{eq:sharedcoverage}
\end{equation}

For fixed $K$, suppose $v_c=v$, $\kappa_c=\kappa$,
$m_c/m\to\pi_c>0$, and
$\delta_c=\delta_0+\xi_c/\sqrt m$ with
$\sum_c\pi_c\xi_c=0$.
The gaps $\Delta_m^{\rm separate}=(1-\alpha)-C_{\rm separate}$
and $\Delta_m^{\rm shared}=(1-\alpha)-C_{\rm shared}$ satisfy
\eqref{eq:poolingsummary}. Subtracting them gives
\begin{equation}
 \Delta_m^{\rm shared}-\Delta_m^{\rm separate}
 =\frac{z_\alpha\phi(z_\alpha)}m
 \left\{\frac{\sum_c\pi_c\xi_c^2}{v(1+\kappa)}
             -\frac{K-1}{\kappa}\right\}+o(m^{-1}).
 \label{eq:poolingtradeoff}
\end{equation}
\end{proposition}

In \eqref{eq:sharedcoverage}, $b_c$ measures the bias from
imposing a common change, while $\omega_c$ measures uncertainty
in the estimated shared change, both on the predictive
residual scale. The local comparison in
\eqref{eq:poolingtradeoff} considers differences between
contexts of the same order as estimation error. Separate
fitting helps when the first term, measuring these differences,
exceeds the additional estimation cost of fitting $K$
parameters instead of one. Under these local alternatives,
$\Var_Q(\delta_C)=\sum_c\pi_c\xi_c^2/m$, and
$\Delta_m^{\rm shared}-\Delta_m^{\rm separate}
=C_{\rm separate}-C_{\rm shared}$.
Thus \eqref{eq:poolingtradeoff} compares coverage as well as
the gaps displayed separately in \eqref{eq:poolingsummary}.

\begin{proof}
Write $\sigma_{Z,c}^2=\lambda_c^2v_c+\tau_c^2$.
Conditional on the context counts, the separate mean-change
estimates are independent and satisfy
\[
 \widehat\delta_c
 =\frac{\overline Z_c^{\rm t}-a_c}{\lambda_c}-\mu_c^P,
 \qquad
 \widehat\delta_c\sim\N(\delta_c,J_c^{-1}).
\]
Their precision-weighted average is the shared maximum-likelihood
estimate, with distribution
$\N(\overline\delta,v_\delta)$.

The fitted predictive variance in context $c$ is
$v_{{\rm pred},c}=v_c/(1+\kappa_c)$. Under the true target distribution,
conditional on the shared estimate, the standardized predictive
residual has variance one and mean
$(\delta_c-\widehat\delta)/\sqrt{v_c(1+\kappa_c)}$.
Averaging over the fitting sample gives a normal residual with
mean $b_c$ and variance $1+\omega_c$. Under the fitted target
distribution, the score is absolute standard normal in every
context, so the population weighted threshold is $z_\alpha$.
Evaluating the true residual probability between
$-z_\alpha$ and $z_\alpha$, and averaging over contexts,
proves \eqref{eq:sharedcoverage}.

Near $b=\omega=0$, the normal interval probability expands as
$1-\alpha-z_\alpha\phi(z_\alpha)(b^2+\omega)+o(b^2+\omega)$.
In the homogeneous model,
\[
 \omega_c=\frac1{m\kappa},\qquad
 b_c=\frac{\xi_c}{\sqrt{mv(1+\kappa)}}+o(m^{-1/2}).
\]
Averaging $b_c^2+\omega_c$ with weights $\pi_c$ gives the shared
gap in \eqref{eq:poolingsummary}, since
$\Var_Q(\delta_C)=\sum_c\pi_c\xi_c^2/m$.
The separate-fit formula \eqref{eq:contextlocationcost}
gives $\Delta_m^{\rm separate}
=z_\alpha\phi(z_\alpha)K/(m\kappa)+o(m^{-1})$.
Subtracting this gap from the shared gap proves
\eqref{eq:poolingtradeoff}.

For multinomial context counts, the same first-order terms
follow after averaging. The reciprocal-count expansion in
\Cref{app:contextcostproof} gives the separate gap. For the
shared fit, $\omega_c=1/(m\kappa)$ whenever all contexts are fitted,
while $\overline\delta-\sum_c\pi_c\delta_c$ has variance
$\Var_Q(\delta_C)/m=O(m^{-2})$ under the local alternatives.
This affects the shared gap only beyond order $m^{-1}$.
Counts too small to fit have exponentially small probability.
The prespecified fallback in \Cref{cor:contextcost} therefore
does not change either leading expansion.

\end{proof}

\section{Experimental implementation}
\label{app:protocols}

The Gaussian comparison, pooling study, and Air Quality study use
seeds $20261002$, $20261003$, and $20261004$, respectively.

\subsection{Gaussian comparisons and non-Gaussian check}
\label{app:gaussianprotocol}

\paragraph{Source fitting and corrections.}
Outcome means and unbiased variances, sensor intercepts, and
loadings are fitted by context. With $N_c$ training observations,
sensor residual variances use $N_c-2$ degrees of freedom. Context frequencies use half-count
smoothing. Separate corrections invert fitted sensor regressions at
target sensor means and subtract source outcome means. Shared corrections
average these changes with precisions
$m_c\widehat\lambda_c^2/
(\widehat\lambda_c^2\widehat v_c+\widehat\tau_c^2)$.

\paragraph{Comparators.}
Covariate weights use logistic regression on $(c,z,z^2,cz,cz^2)$,
source-based standardization, and $L_2$ inverse penalty strength
one. Odds are adjusted for sampling proportions, and log ratios are
clipped to $[-8,8]$. Pooled label adjustment fits a pooled sensor
model and a two-component target outcome mixture by sensor
likelihood, predicting from $Z$ alone. Its ratio denominator is
the mixture of fitted source outcome densities weighted by fitted
source context frequencies. L-BFGS-B retains the best of four
initializations, using default stopping tolerances, means in
$[-5,5]$, variances in $[0.05,10]$, and mixture log odds in
$[-3,3]$. Coverage evaluation uses exact normal probabilities for
linear predictors and $120$ Gauss--Hermite nodes per target
component for the pooled predictor.

\paragraph{Non-Gaussian check.}
The $150$ paired fits use seed $202610020729$. Gaussian maximum
likelihood in context zero and a two-component Gaussian mixture
with common component variance and five EM starts in context one
are combined with a shared linear sensor regression. Separate
tilts in $[-1.5,1.5]$ maximize target-input likelihood. Context
frequencies use half-count smoothing. Covariate weights are
$\widehat q_C(c)\widehat q(z\mid c)/
\{\widehat p_C(c)\widehat p(z\mid c)\}$. Both input densities
integrate the fitted sensor model against the fitted outcome density,
with the target using the same estimated tilts as mechanism
adjustment. Covariate scores are standardized source residuals.
Evaluation integrates conditional Gaussian-mixture coverage and
length using $1024$ Gauss--Hermite nodes per true target sensor
component. Doubling the nodes changed each metric by less than
$10^{-6}$ in replication zero. SEs are replication sample standard
deviations divided by $\sqrt{150}$.

\subsection{Sharing corrections and finite calibration}
\label{app:fittingprotocol}\label{app:poolingprotocol}

Centering and scaling $c$ under $\pi_c\propto0.8^c$ gives $u_c$ in
\Cref{fig:poolingcomparison}. Source contexts are uniform, and source
parameters are known or fitted from training data. Corrections
follow the rules above, with empty target contexts receiving the shared
change. Both methods retain context and use standardized residuals
and FWQ. Formula curves use $m_c=m\pi_c$, while simulations use sampled counts.

The scalar study uses target mean one, variance bounds $[0.05,10]$,
and $1000$ independent normal and chi-square fitting pivots, with
calibration samples nested over $n$. For each fit, subtract exact
population-quantile coverage from empirical-quantile coverage.
Adding the mean difference to \eqref{eq:scaleexact} gives the
control-variate estimates in \Cref{sec:finitecalibration}.
Fits with infinite ratio second moments are retained.

\subsection{Air Quality data and evaluation}
\label{app:airprotocol}

\paragraph{Data and source models.}
The outcome is log \texttt{CO(GT)}. The five log-transformed sensors are
\texttt{PT08.S1(CO)}, \texttt{PT08.S2(NMHC)}, \texttt{PT08.S3(NOx)},
\texttt{PT08.S4(NO2)}, \texttt{PT08.S5(O3)}. Exclude records
coded $-200$ in these variables, \texttt{T}, \texttt{RH}, or
\texttt{AH}, and require positive CO. Retained sensors are positive.
Missingness elsewhere does not determine inclusion. Contexts use
$T\ge20^\circ$C for warm and $\mathrm{RH}\ge50\%$ for humid.
Absolute humidity \texttt{AH} is not a predictor.
Source outcomes have context-specific means and scales and a shared
standardized residual mixture selected by BIC from one to four
Gaussian components. The sensor model is
\[
 Z=a_c+B_c(Y,Y^2)^\top+\varepsilon_c,
 \qquad\varepsilon_c\sim\N(0,\Sigma_c),
\]
with covariance shrunk $2\%$ toward its diagonal. Contexts with
fewer than $20$ training records share a fit within their
temperature group, or use all training records if necessary.
Context frequencies use half-count smoothing.

\paragraph{Adaptation and calibration.}
Each context has a tilt coefficient in $[-3,3]$, fitted by
L-BFGS-B from zero with default tolerances. The feature is
$g(y)=\operatorname{clip}(y,\log0.3,\log8)-\log2$.
Integrals use $271$ equally weighted log-CO grid points on
$[-4.5,4.5]$. Tilt normalizers include grid spacing, posterior
weights sum to one, and predictive variances are floored at $10^{-4}$.
Covariate-S uses context indicators and within-context linear and
squared sensor features, with the source scaling, logistic penalty,
odds correction, and clipping specified above. Both methods use
standardized residuals and FWQ. Standard CP uses the source score
and an unweighted split-conformal quantile.

\paragraph{Controlled selection.}
Exclude $35$ design days and randomly allocate the remaining days
in proportions $40/20/20/20$ to source training/calibration and target
fitting/evaluation. Source/target records are retained with
probabilities $1-\operatorname{expit}(\xi)$ and
$\operatorname{expit}(\xi)$, where
$\operatorname{expit}(\xi)=1/(1+e^{-\xi})$. Their density ratio is
proportional to $e^\xi$. The three designs use
\[
 \xi=0,\qquad \xi=0.8s_H+0.4s_T,\qquad
 \xi=0.25s_T+1.2s_H\ell(y),
\]
where $\ell$ is defined in \Cref{sec:air} and $s_T,s_H$ are $1$
for warm/humid conditions and $-1$ otherwise.

\paragraph{Calendar-based splits and evaluation.}
Transfers use weekdays to weekends, March--May to
September--November 2004, and March 10--April 30, 2004
($960$ hours) to March 1--April 4, 2005 ($808$ hours).
These source and target groups use all eligible days, including design days.
Source days are split $60/40$ for training/calibration, and target
days are split $50/50$ for fitting/evaluation. The labeled reference
keeps the source predictor and calibrates on target fitting outcomes.
Metrics average evaluation hours within splits, then equally over
$50$ splits. Paired SEs use the sample standard deviation of
within-split differences divided by $\sqrt{50}$. Day splits separate
sample roles, without establishing temporal independence or
independent validation of design choices.

\end{document}